\documentclass[11pt, a4paper]{amsart}
\usepackage{amsfonts}
\usepackage{amssymb}
\usepackage[dvips]{graphics}
\usepackage{epsfig}
\usepackage{euscript}
\usepackage{color}

\usepackage{fancyvrb}

\usepackage[pdftex]{hyperref}

 \newtheorem{thm}{Theorem}[section]
 \newtheorem{cor}[thm]{Corollary}
 \newtheorem{lemma}[thm]{Lemma}
 \newtheorem{prop}[thm]{Proposition}
 \theoremstyle{definition}

 \newtheorem{assumption}[thm]{Assumption}
 \numberwithin{equation}{section}

\newcommand{\caA}{{\mathcal A}}
\newcommand{\caB}{{\mathcal B}}
\newcommand{\caC}{{\mathcal C}}

\newcommand{\caE}{{\mathcal E}}
\newcommand{\caF}{{\mathcal F}}

\newcommand{\caH}{{\mathcal H}}
\newcommand{\caI}{{\mathcal I}}

\newcommand{\caL}{{\mathcal L}}

\newcommand{\caO}{{\mathcal O}}
\newcommand{\caP}{{\mathcal P}}

\newcommand{\caS}{{\mathcal S}}

\newcommand{\bbC}{{\mathbb C}}

\newcommand{\bbI}{{\mathbb I}}

\newcommand{\bbN}{{\mathbb N}}

\newcommand{\bbR}{{\mathbb R}}

\newcommand{\bbZ}{{\mathbb Z}}

\newcommand{\iu}{\mathrm{i}}
\newcommand{\id}{\bbI}

\newcommand{\str}{^{*}}
\newcommand{\ep}[1]{\mathrm{e}^{#1}}

\newcommand{\dd}{\,\mathrm{d}}

\newcommand{\Tr}{\mathrm{Tr}}

\newcommand{\norm}[1]{\left\Vert #1 \right\Vert}
\newcommand{\ad}[1]{\mathrm{ad}_{#1}}

\newcommand{\supp}{\mathrm{supp}}

\newcommand{\be}{\begin{equation}}
\newcommand{\ee}{\end{equation}}
\newcommand{\bea}{\begin{eqnarray}}
\newcommand{\eea}{\end{eqnarray}}
\newcommand{\beann}{\begin{eqnarray*}}
\newcommand{\eeann}{\end{eqnarray*}}

\begin{document}

\title{The adiabatic theorem and linear response theory for extended quantum systems}

\author{Sven Bachmann}
\address{Mathematisches Institut der Universit{\"a}t M{\"u}nchen \\ 80333 M{\"u}nchen \\ Germany}
\curraddr{Department of Mathematics \\ University of British Columbia \\ Vancouver, BC V6T 1Z2 \\ Canada}
\email{sbach@math.ubc.ca}

\author{Wojciech de Roeck}
\address{ Instituut Theoretische Fysica, KULeuven  \\
3001 Leuven  \\ Belgium }
\email{wojciech.deroeck@kuleuven.be}

\author{Martin Fraas}
\address{ Instituut Theoretische Fysica, KULeuven  \\
3001 Leuven  \\ Belgium }
\curraddr{Department of Mathematics \\ Virginia Tech \\ Blacksburg, VA 24061-0123 \\ USA}
\email{fraas@vt.edu}

\date{\today }

\begin{abstract}
The adiabatic theorem refers to a setup where an evolution equation contains a time-dependent parameter whose change is very slow, measured by a vanishing parameter $\epsilon$. Under suitable assumptions the solution of the time-inhomogenous equation stays close to an instantaneous fixpoint. In the present paper, we prove an adiabatic theorem with an error bound that is independent of the number of degrees of freedom. Our setup is that of quantum spin systems where the manifold of ground states is separated from the rest of the spectrum by a spectral gap.
One important application is the proof of the validity of linear response theory for such extended, genuinely interacting systems. In general, this is a long-standing mathematical problem, which can be solved in the present particular case of a gapped system, relevant e.g.~for the integer quantum Hall effect.
\end{abstract}

\maketitle


\section{Introduction}\label{sec:intro}

\subsection{Adiabatic Theorems}
The \emph{adiabatic theorem} stands for a rather broad principle that can be described as follows. Consider an equation giving the evolution in time $t$ of some $\varphi(t) \in B$, with $B$ some Banach space:
\be \label{eq: first ad eq}
\frac{d}{dt} \varphi(t)=  L(\varphi(t),\alpha_t),
\ee
where $L: B\times \bbR \to B$ is a smooth function and $\alpha_t$ is a parameter depending parametrically on time. Now, assume that for $\alpha_t=\alpha$ frozen in time, the equation exhibits some tendency towards a fixpoint $\varphi_\alpha$, as expressed, for example, by convergence in Cesaro mean  $\frac{1}{T}\int_0^T\varphi(t)dt   \to \varphi_\alpha$, as $T\to \infty$ for any initial $\varphi(0)$. Intuitively, one can then expect that for a very slowly $t$-dependent $\alpha_t$, the solution $\varphi(t)$ shadows the instantaneous fixpoint $\varphi_{\alpha_t}$ if the initial condition $\varphi(0)$ is close to $\varphi_{\alpha_0}$.
 
 There is a multitude of results proving such a principle, starting from the old works~\cite{BornFock,Kato50,kasuga1961adiabatic}, see \cite{Teufel,Panati2003,AE99,Jansen} for more recent accounts focusing on quantum dynamics,  \cite{lorenz2005adiabatic} for a version outside the quantum formalism, \cite{bradford2011adiabatic} for a version tailor-made for Markov processes, and \cite{gang2015adiabatic} for a case where $L_\alpha(\cdot) := L(\cdot, \alpha): B\to B$ is not linear.
 In order to guide the discussion, let us state explicitly the very general result of~\cite{AFGG} that is however restricted to a linear $L_\alpha$. Introducing a scaling parameter $\epsilon$ and a rescaled time $s=\epsilon t$, we use $s \in [0,1]$ as the basic variable and an $s$-dependent parameter $\alpha_s$, setting $\phi(s) = \varphi(t)$. The above equation \eqref{eq: first ad eq} then turns into
\be \label{eq: first ad linear}
\epsilon\frac{d}{ds} \phi(s)=  L_{\alpha_s}\phi(s).
\ee
In~\cite{AFGG}, $L_{\alpha}$ is assumed to generate a contractive semigroup for any $\alpha$, and to be such that $B=\mathrm{Ker}L_{\alpha} \oplus \overline{\mathrm{Ran}L_{\alpha}}$. This corresponds to the intuition above, insofar as the latter condition is equivalent to mean-ergodicity 
\be \label{eq: mean ergodic theorem}
\frac{1}{T}\int_0^T e^{tL_\alpha}\phi  \dd t\longrightarrow \caP_\alpha\phi,
\ee
for all $\phi\in B$, where $\caP_\alpha$ the ergodic projection on $\mathrm{Ker}L_{\alpha}$.
 The result of \cite{AFGG} is then that 
\be \label{eq: result afgg}
\sup_{s \in [0,1]}  \left\Vert (1-\caP_{\alpha_s})\phi(s) \right\Vert \leq C \epsilon 
\ee
where $\phi(s)$ is a solution of~(\ref{eq: first ad linear}), with an initial condition $\phi(0)$ satisfying the same bound. In other words: the dynamically evolved $\phi(s)$ remains within the instantaneous kernel, up to a small leak. 

\subsection{Fundamental applications}

There are at least two domains in physics where the adiabatic principle is fundamental. 

\subsubsection*{Thermodynamic Processes}
   We consider a phase space $\Omega$ (generalised positions and momenta) of a physical system, endowed with a symplectic structure and a Liouville measure $\omega$. A Hamilton function $H_\alpha:\Omega\to\bbR$ generates a flow on $\Omega$, for which the measure $\omega$ is invariant. Let us choose $B$ to be the class of signed measures that are absolutely continuous w.r.t.\ $\omega$ and we identify these measures with their densities $\varphi$, so we can put $B=L^1(\Omega,\omega)$.  Then the flow on $\Omega$ is naturally lifted to $B$ as 
$$
\frac{d}{dt}\varphi(t) = \{H_{\alpha},\varphi(t) \}
$$
where $\{\cdot,\cdot \}$ is the Poisson bracket corresponding to the symplectic structure.
 If the flow generated by the Hamiltonian $H_\alpha$ is ergodic on any energy shell $\Omega_E=\{ x \in \Omega: H(x)=E \}$ and some technical conditions are met, then~\eqref{eq: mean ergodic theorem} is  satisfied with $P_\alpha$ projecting on densities that are functions of $H_\alpha$. Rescaling the time-dependent equation for the family $H_{\alpha_{t}}$ to get
\begin{equation*}
\epsilon \frac{d}{ds}\phi(s) = \{H_{\alpha_s},\phi(s)\},
\end{equation*}
we are in the framework described above and the result \eqref{eq: result afgg} applies. It teaches us that the evolution proceeds via instantaneous fixpoints, i.e.\ densities that are functions of $H_{\alpha}$.  
Indeed, it is a basic tenet of thermodynamics that the evolution corresponding to a sufficiently slow $\alpha_t$ --- called in this context a \emph{quasi-static} process --- remains within the set of \emph{Gibbs states}. 
\begin{equation} \label{eq: gibbs states}
\varphi_{\alpha,\beta} \propto \ep{-\beta H_\alpha},
\end{equation}
where $0<\beta<\infty$. Of course, the fixpoints are not automatically of the form \eqref{eq: gibbs states} but typicality considerations involving `equivalence of ensembles' make them physically equivalent to Gibbs states, in the case where the system has a large number of degrees of freedom (see also below).  Therefore, the adiabatic principle is often invoked to justify the relevance and prevalence of Gibbs state to describe time-inhomogenous systems.

\subsubsection*{Quantum Dynamics} Here the Banach space $B$ is a separable Hilbert space $\caH$, and there is a family of self-adjoint operators $H_\alpha$ (`Hamiltonians') acting on $\caH$.  The dynamics of $\phi \in \caH$ in rescaled time is given by the driven Schr\"odinger equation
\be \label{eq: first schro}
\epsilon\frac{d}{ds} \phi(s) = -\iu H_{\alpha_s} \phi(s).
\ee 
In this case, the condition \eqref{eq: mean ergodic theorem} is satisfied by the spectral theorem. To exclude uninteresting trivial cases, we assume that $\mathrm{Ker} H_\alpha$ is non-empty for all $\alpha$. For simplicity, we also assume that $H_\alpha$ are bounded. Then, the result is given again by \eqref{eq: result afgg}, which corresponds in this case to the original works of Born, Fock and Kato~\cite{BornFock,Kato50,kasuga1961adiabatic}.

Note that despite the formulation, there is no distinguished role played by the eigenvalue $0$ of $H_\alpha$. 
Indeed, assume that $H_\alpha$ has an eigenvalue $E_{\alpha}$, depending smoothly on $\alpha$ (recall that we assume everything to be smooth throughout). If $\phi(s)$ solves \eqref{eq: first schro}, then 
$$
\widetilde{\phi}(s)=e^{-(\iu/\epsilon)\int_0^s E_{\alpha_u} d u}\widetilde\phi(s)
$$
solves 
\be \label{eq: second schro}
\epsilon\frac{d}{ds} \widetilde{\phi}(s) = -\iu \widetilde H_{\alpha_s}\widetilde{\phi}(s),
\ee 
with $\widetilde H_{\alpha}=H_\alpha-E_\alpha$. The latter equation \eqref{eq: second schro} reduces the evolution problem for $E_\alpha$-eigenvectors to the case above describing the evolution of vectors in the kernel. 
The adiabatic theorem in quantum mechanics is often invoked to argue that a system can be prepared in a certain state by switching the Hamiltonian slowly.  

Finally, it is worthwhile to note here that the adiabatic theorem can be substantially strengthened if the $H_\alpha$ have a \emph{spectral gap}, i.e.\ if
the spectra $\sigma(H_\alpha)\setminus \{0\}$ are bounded away from $0$, uniformly in $\alpha$. In that case, at times in which the first $m$-derivatives of the Hamiltonian vanish, the result \eqref{eq: result afgg} holds with an improved error bound $C_m\epsilon^m$, where the integer $m$ depends on the smoothness in $\alpha$ and it can be made arbitrarily large if $\alpha\mapsto H_\alpha$ is $C^\infty$. 

\subsection{Extended systems}

The two above applications of the adiabatic theorem --- thermodynamic processes and quantum dynamics --- refer to dynamical systems of physics. A key parameter of any such dynamical system is the \emph{number $N$ of degrees of freedom} it has. A related important notion is that of the spatial extension of the physical system and its \emph{locality}. Indeed, there are mathematical results that become irrelevant when the number of degrees of freedom grows large. Take for example the KAM theorem for Hamiltonian dynamics, stating that small perturbations of integrable systems inherit the integrability, see e.g.~\cite{Broer:2004vs} and references therein. How small the perturbation must be depends on $N$ and it must be vanishing with growing $N$. It is therefore commonly believed that the KAM theorem has little or no bearing on large physical systems. 
On the face of it, the adiabatic theorem suffers from a similar hitch, as the constant $C$ in \eqref{eq: result afgg} grows typically linearly  with $N$.  As we will explain below, this divergence is unavoidable as the adiabatic theorem with an $N$-independent $C$ is shown to be wrong.

Yet, a lot of physics applications do involve \emph{macroscopic} systems, where $N \sim 10^{23}$. Indeed, thermodynamics always assumes that the dynamical system is macroscopic, namely that one should take the `thermodynamic limit' $N\to\infty$ which justifies the notion of typicality or the applicability of laws of large numbers.  In quantum dynamics, many considerations involving the adiabatic theorem deal with extended lattices of spins. We mention most notably the problem of the classification of gapped phases or the dynamical crossing of a quantum critical point, both notions being meaningless for small systems where $N$ is fixed.
 
Hence there is a fundamental and practical need to find a form of the adiabatic theorem that is applicable for arbitrarily large systems. This is precisely what we provide in the present article for gapped quantum spin systems. There are two ingredients that will allow us to obtain again~\eqref{eq: result afgg} with a $N$-independent constant $C$ for macroscopic systems.  First and foremost, we will exploit the locality of the system: only nearby degrees of freedom interact directly, and the quantum dynamics exhibits a finite speed of propagation. Second, 
 we will choose a metric that probes $\phi$ only through a few degrees of freedom at a time, instead of the norm on $B$ which considers all degrees of freedom simultaneously.
 
 \subsection{Linear response theory}
 
With an adiabatic theorem for extended systems at hand, we shall prove the validity of linear response theory. Initially developed by Kubo in~\cite{Kubo:1957cl}, the very generally phrased theory deals with the response of a physical system to a driving, in first order in the strength of the driving. In a typical setting, the driving (for example an electromotive force) is switched on adiabatically from $0$ in the infinite past to reach an intensity $\alpha>0$ at $t=0$, where one computes the induced response (the size of an electrical current in the example) to first order in $\alpha$. This yields a linear relationship between response and driving, with a proportionality constant --- the response coefficient --- given by the so-called Kubo formula.

Here again, the theory is well established for small systems whereas for macroscopic systems a proof of its validity is considered of fundamental interest, see~\cite{Simon:1984aa}. The crucial issue is that of the order of the limits: First, the system must be taken to the macroscopic limit $N\to\infty$, the adiabatic limit $\epsilon \to 0$ must be taken second, while the limit of linear response $\alpha\to 0$ is carried out last. We shall prove (i) that these limits exist in this order and (ii) give an explicit formula for the linear response coefficient, which is formally equal to Kubo's formula but rigorously well-defined. 
 

\section{Setup and Results}\label{sec:Results}

\subsection{Quantum spin systems}\label{sub:Setup}
To introduce the spatial extent and the local structure of our system, we choose a countable graph $\Gamma$ endowed with the graph metric that we denote by $\mathrm{dist}(\cdot,\cdot)$.  We assume $\Gamma$
to be  $d$-dimensional in the sense that there is a {$\kappa<\infty$} such that
\begin{equation}\label{ddim}
\sup_{y\in\Gamma}\vert \{ x\in\Gamma: \mathrm{dist}(x,y) \leq r \}\vert\leq {\kappa} r^d.
\end{equation}
For simplicity, one can keep in mind the main example $\Gamma=\mathbb{Z}^d$.
 To each vertex $x\in\Gamma$, we associate a quantum system Hilbert space $\caH_x$, where $\sup\{\mathrm{dim}(\caH_x):x\in\Gamma\}<\infty$.  Let  $\caF(\Gamma)$ denote the set of finite subsets of $\Gamma$. For any $\Lambda\in\caF(\Gamma)$, we write $\vert\Lambda\vert$ for its volume (cardinality) and 
 $\mathrm{diam}(\Lambda) = \max_{x,y\in\Lambda}\mathrm{dist}(x,y)$. We further define its corresponding Hilbert space $\caH^\Lambda:= \otimes_{x\in\Lambda}\caH_x$ and the algebra of observables $\caA^\Lambda:= \caB(\caH^\Lambda)$. For $\Lambda' \subset \Lambda$, $\caA^{\Lambda'}$ is in a natural way isomorphic to a subalgebra of $\caA^{\Lambda}$ of the form $\caA^{\Lambda'} \otimes 1_{\Lambda\setminus \Lambda'}$ and we use this embedding without comment.  For any operator $O$, we write $\supp(O)$ for the smallest $\Lambda$ such that $O \in \caA^\Lambda$. Each $\caA^\Lambda$ is equipped with the trace $\Tr_\Lambda: \caA^\Lambda\to \bbC$.
 \subsection{Dynamics}
To define the dynamics, it is useful to first consider \emph{interactions}. An interaction $\Phi$ is a {map that associates an observable $\Phi(X) = \Phi(X)\str \in \caA^X$ to any finite set $X\in\caF(\Gamma)$.} It defines a Hamiltonin $H^\Lambda  \in \caA^\Lambda$, for any $\Lambda \in \caF(\Gamma)$ by setting
\begin{equation*}
H^\Lambda:= \sum_{X\subset\Lambda}\Phi(X).
\end{equation*}
The notion of an interaction will allow us $(i)$ to consider {all} finite volumes $\Lambda$ at the same time and $(ii)$ to naturally put a locality assumption by requiring that $\Vert \Phi(X)\Vert$ decays when $\mathrm{diam}(X)$ grows. 

Below, we will {systematically} drop the superscripts $\Lambda$ as all assumptions, bounds and statements apply uniformly in $\Lambda$.

{As discussed in the introduction, we shall} consider a time-dependent interaction $\Phi_{s}$, and the associated {time-dependent Hamiltonians} $H_s$, for $s\in[0,1]$. The evolution on the Hilbert space $\caH$ is now given by the time-dependent Schr\"odinger equation 
\begin{equation}\label{SchrödU}
\iu \epsilon \frac{d}{ds}\psi_\epsilon(s) = H_s \psi_\epsilon(s),
\end{equation}
where the parameter $\epsilon$ {indicates} that $s$ should be thought of as a rescaled time, cf.\ the introduction.

The setup described above is standard in quantum spin systems, see e.g.~Section~6.2 of~\cite{Bratteli:1997aa}, and it will be elaborated further in Section~\ref{sec: infinite volume setup}. It {includes all} {standard lattice models studied in the condensed matter theory. We provide some examples in Section~\ref{sec:examples}.}

\subsection{{Assumptions}}

{The adiabatic theorem we prove below holds for time dependent Hamiltonians of quantum spin systems that are \emph{gapped} and depend \emph{smoothly on time}. We formulate these assumptions precisely here, while keeping the notational burden to a minimum.}
\begin{assumption}[Gap Assumption]\label{A:patch}
For any $s\in[0,1]$, let $\Sigma_s$ be the spectrum of the Hamiltonian $H_s$. There is a splitting $\Sigma_s = \Sigma^{1}_s\cup \Sigma^{2}_s$, such that $[\min\Sigma^{1}_s,\max\Sigma^{1}_s] \cap \Sigma^{2}_s = \emptyset$ and
\begin{equation*}
\gamma:= \inf\{\mathrm{dist}(\Sigma^{1}_s, \Sigma^{2}_s):s\in[0,1]\}>0,
\end{equation*}
uniformly in the volume $\Lambda$.
\end{assumption}
\noindent The condition $[\min\Sigma^{1}_s,\max\Sigma^{1}_s] \cap \Sigma^{2}_s = \emptyset$ simply excludes that $\Sigma^{1,2}_s$ would consist of several interlaced patches, and the min/max exist because the spectrum is a closed set.

{Whenever Assumption~\ref{A:patch} holds, we} define $P_s$ to be the spectral projection of the Hamiltonian $H_s$ associated to $\Sigma^{1}_s$.

{The second assumption addresses the smoothness of the interactions on the one hand, and their locality on the other hand. The latter is expressed in terms of a class $\caB_{\caE,\infty}$ whose explicit definition will be given later, and whose meaning is of an exponential decay of $\norm{ \Phi_s(X)}$ with $\mathrm{diam}(X)$ for all $s\in[0,1]$}.
\begin{assumption}[locality and differentiability]\label{A:Hamiltonians}
Let $m$ be an integer not smaller than $1$.
For any $X \in \caF(\Gamma)$, the matrix-valued function $s\mapsto \Phi_s(X)$ is $(d+m)$-times continuously differentiable. These derivatives define time-dependent interactions $\{\Phi_s^{(k)}(X):X \in \caF(\Gamma)\}$ such that $\Phi^{(k)}\in\caB_{\caE,\infty}$ for $0\leq k\leq d+m$, where $\Phi_s^{(0)}=\Phi_s$. Moreover, $\Phi_s^{(k)}(X)\vert_{s=0} = 0$ for all $X\in\caF(\Gamma)$ and $1\leq k\leq d+m$.
\end{assumption}

Note that if Assumption~\ref{A:Hamiltonians} holds, then by standard perturbation theory, $s\mapsto\min\Sigma^{1}_s$ and $s\mapsto\max\Sigma^{1}_s$ defined in Assumption~\ref{A:patch} are continuous, although not necessarily uniformly so in $\Lambda$. Similarly, $s\mapsto P_s$ is $(d+m)$-times continuously differentiable.

For {simplicity, the reader may keep in mind} \emph{finite-range} interactions{: There is $R>0$ such that}
$$
\Phi_s(X) =0 \qquad \text{unless $\mathrm{diam}(X) \leq R$.}
$$
In that case the statement $\Phi^{(k)}\in\caB_{\caE,\infty}$ in Assumption \ref{A:Hamiltonians} reduces to 
$$
\sup\left\{\Vert\Phi_s^{(k)}(X)\Vert:X\in\caF(\Gamma)\right\} <\infty,
$$
{uniformly in $s\in[0,1]$.}
\subsection{Result}
To state our main result, let us consider a normalised solution $\psi_\epsilon(s)$ of the Schr\"odinger equation \eqref{SchrödU} such that the initial condition $\psi_\epsilon(0)$ lies in the spectral patch $\Sigma_0^1$, i.e.
\begin{equation}\label{Initial condition}
P_0\psi_\epsilon(0)=\psi_\epsilon(0),\qquad \Vert\psi_\epsilon(0)\Vert = 1.
\end{equation}
Our result then expresses that in a certain sense the {adiabatically} evolved $\psi_\epsilon(s)$ remains in the instantaneous patch.
\begin{thm}\label{thm:mainVector}
Let Assumption~\ref{A:patch} hold. 
\begin{enumerate}
\item If Assumption~\ref{A:Hamiltonians} is satisfied for $m=1$, then there is a vector $\widetilde \psi_\epsilon(s)\in\mathrm{Ran}(P_s)$ with $\Vert\widetilde \psi_\epsilon(s)\Vert = 1$ and a $C_1<\infty$ that is independent of $\Lambda$ and $\epsilon$, such that
\be
\sup_{s\in[0,1]}\left\vert
\langle \psi_\epsilon(s), O \psi_\epsilon(s) \rangle - \langle \widetilde \psi_\epsilon(s), O \widetilde \psi_\epsilon(s) \rangle
\right\vert \leq C_1 |\supp(O)|^2  \Vert O\Vert \epsilon
\ee
for any $O$ with $\supp(O)\subset \Lambda$.
\item If Assumption~\ref{A:Hamiltonians} is satisfied for $m>1$, and if $\Phi_s^{(k)}(X)\vert_{s=1} = 0$ for all $X\in\caF(\Gamma)$ and $1\leq k\leq d+m$, then the above bound is strengthened at the endpoint $s=1$ to
\be
\sup_{s\in[0,1]}\left\vert
\langle \psi_\epsilon(1), O \psi_\epsilon(1) \rangle - \langle \widetilde \psi_\epsilon(1), O \widetilde \psi_\epsilon(1) \rangle
\right\vert \leq C_m |\supp(O)|^2  \Vert O\Vert \epsilon^m.
\ee
\end{enumerate}
\end{thm}
To compare this result to \eqref{eq: result afgg}, it is convenient to first assume that the spectral patch $\Sigma^1_s =\{0\}$, namely $0$ is an {isolated} eigenvalue for all $s$. Then Theorem \ref{thm:mainVector} in the case $m=1$ is the statement in \eqref{eq: result afgg}, except that the proximity of $\psi_\epsilon(s)$ to the instantaneous spectral subspace $P_s$ is expressed by a coarser topology{, while the bound holds uniformly for all finite volumes}. The presence of the  gap, Assumption \ref{A:patch}, which is anyhow crucial for our result, allows for $\Sigma_s^1$ to contain more spectrum than a single eigenvalue. {Provided enough smoothness, the gap further allows for an error bound $\epsilon^m$ instead of simply $\epsilon$, as was already} studied via the adiabatic expansion, see  \cite{Berry90, Nenciu, Garrido, Hagedorn}. The novelty of our approach, both at the technical and conceptual level, is the construction of an adiabatic expansion that uses the locality of the dynamics and is compatible with it.

A natural question that is not explicitly addressed by the theorem is how to pick the vector $\widetilde \psi_\epsilon(s)\in\mathrm{Ran}(P_s)$. In general, there is not much to say about this, but a special case arises if we assume that the spectral patch $\Sigma^{1}_s$ corresponding to $\mathrm{Ran}(P_s)$ is (nearly) degenerate.  Let us define the splitting
$$
\delta =\sup_{s \in [0,1]}  \left(\max \Sigma^{1}_s-\min \Sigma^{1}_s \right)
$$
and assume near-degeneracy in the form
\begin{equation} \label{eq: condition delta}
\delta \leq C \mathrm{min}(\epsilon^2,  \epsilon/|\Lambda|^{-1})
\end{equation}
 It is indeed very natural in many-body systems that a degeneracy between ground states is lifted very slightly (in many examples the splitting $\delta$ is exponentially small in the volume).
As we  shall show on page~\pageref{eq: equation tilde psi}, 
our expansion yields that, if \eqref{eq: condition delta} holds, then  $\widetilde \psi_\epsilon(s)$ satisfying Theorem~\ref{thm:mainVector}(i) can be simply chosen as a solution of the \emph{parallel transport} equation 
\begin{equation} \label{eq: parallel transport}
P_s \frac{d}{ds}\Omega(s)=0, \qquad  \Omega(0)=\psi_0,
\end{equation}
in which case it is naturally $\epsilon$-independent. Equivalently, $\Omega(s)$ is the solution of the  equation 
\begin{equation} \label{eq: parallel transport 2}
\iu\frac{d}{ds}\Omega(s) =\iu [P_s, \dot P_s]\Omega(s), \qquad  \Omega(0)=\psi_0,
\end{equation}
even though the generator $[P_s, \dot P_s]$ is not a local Hamiltonian.
Choices of $\widetilde \psi_\epsilon(s)$ that fulfill Theorem~\ref{thm:mainVector}(ii) can be constructed as well by continuing the expansion to higher orders and strenghtening \eqref{eq: condition delta}, see \cite{TeufelAd}.

\subsection{Examples}
\label{sec:examples}

We now discuss the result further via some examples, and we shall illustrate some of the points raised in the introduction.

\subsubsection{Non-interacting spins}\label{sec: non interacting spins}
{Take the interactions $\Phi_s$ to be such that $\Phi_s(X)=0$ whenever $|X|>1$ and write simply $h_{x,s}=\Phi_s(\{x\})$. Hence we have
$$
H_s=\sum_{x \in \Lambda} h_{x,s}.
$$
Then  the class of product vectors $\otimes_{x \in \Lambda} \psi_{x}$ is preserved {by the time evolution} and $\otimes_{x \in \Lambda} \psi_{x,\epsilon}(s)$ is a solution {of the many-body dynamics} provided that for each $x$, $\psi_{x,\epsilon}(s)$ solves the one-site Schr\"odinger equation
$$
\iu \epsilon \frac{d}{ds} \psi_{x,\epsilon}(s)= h_{x,s} \psi_{x,\epsilon}(s).
$$
Let us now assume that  $h_{x,s}$ has an eigenvalue $0$ at the bottom of the spectrum of $h_{x,s}$, and let $P_{x,s}$ be the associated spectral projection. As always, we assume also that the initial condition $\otimes_x \psi_{x,\epsilon}(0)$ satisfies \eqref{Initial condition}, i.e.\ $P_{x,0}\psi_{x,\epsilon}{(0)}=\psi_{x,\epsilon}(0)$ and  $\Vert \psi_{x,\epsilon}{(0)} \Vert =1 $. Writing $P_s = \otimes_{x \in \Lambda} P_{x,s}$, we have than
\be
\label{eq:nis1}
||(1-P_s) \otimes_{x \in \Lambda} \psi_{x{,\epsilon}}(s)||^2 = 1 - \prod_{x \in  \Lambda} || P_{x,s} \psi_{x{,\epsilon}}(s) ||^2. 
\ee
Assuming that the adiabatic theorem holds {in the form~\eqref{eq: result afgg}} at each site $x$ we get 
$$
\lim_{\varepsilon \to 0} ||(1-P_s) \otimes_{x \in \Lambda}  \psi_{x{,\epsilon}}(s)|| = 0.
$$
On the other hand it should be expected - and it is easy to find examples where it is the case -- that $\sup_x ||P_{x,s} \psi_{x,\epsilon}(s)|| \leq 1- c \epsilon$ for some $c >0$, i.e.\ the error term does not vanish completely.  In that case the product on the right hand side of \eqref{eq:nis1} vanishes as $(1-c\epsilon)^{2|\Lambda|}$ with increasing volume $|\Lambda|$, and we get
$$
\lim_{|\Lambda| \to \infty} ||(1-P_s) \otimes_{x \in \Lambda}  \psi_{x{,\epsilon}}(s)|| = 1.
$$
This shows that the adiabatic theorem cannot hold in the form~\eqref{eq: result afgg} with an error bound that is uniform in $\vert \Lambda\vert$. 
}

In the non-interacting context of this example however, Theorem~\ref{thm:mainVector} provides the natural adiabatic statement to be expected. Although a norm bound clearly fails, the solution is provided by probing the adiabatically evolved state only on local observables $O$, with the diabatic error depending in particular on the support of $O$. As should be expected the true difficulty is to deal with interactions, generating a dynamics that does not simply factorise as it did here. Controlling the propagation properties of the interacting dynamics will be crucial.

\subsubsection{Perturbations around gapped systems}
Concrete and --- unlike the above --- non-trivial examples for the applicability of the main theorem arise from perturbations of simple gapped spin systems. 

Indeed, the sole assumption of the main theorem that can not easily be checked by inspection is the presence of a spectral gap separating a spectral patch from the rest of the spectrum. The most natural situation is where the isolated patch lies at the bottom of the spectrum, in which case one can call this patch the ground state space. In many interesting cases the dimension of the ground state space remains bounded and the width of the ground state `band' shrinks to zero when volume grows. In fact, there is a not precisely formulated conjecture in condensed matter that `generic' local Hamiltonians have a single ground state separated by a gap, while it was recently proved that the problem of determining whether a given Hamiltonian is gapped or not is undecidable~\cite{Cubitt:2015ch}. Among the rigorous tools to prove gaps above the ground state space, the martingale method~\cite{Nachtergaele:1996vc} stands out. 

Furthermore, a spectral gap above the ground state energy can be proved for weak perturbations of certain classes of gapped Hamiltonians. Possible choices for the unperturbed $H_0$ include the non-interacting examples from Section \ref{sec: non interacting spins}, more interestingly spin systems with finitely correlated ground states in one dimension~\cite{yarotsky2004perturbations} such as the AKLT model~\cite{AKLT}, and more generally frustration-free, topologically ordered systems~\cite{Bravyi:2011ea,Michalakis:2013gh,Szehr:2015fn} such as the toric code model~\cite{Kitaev:2003ul}. The full Hamiltonian is then of the form
\begin{equation}\label{Perturbation}
H_s = H_0 + \alpha G_s,\qquad s\in[0,1],\alpha\in\bbR,
\end{equation}
with $G_s$ a local Hamiltonian satisfying our smoothness assumption, and $\vert \alpha \vert$ small. Hence, in these cases, all assumptions of Theorem~\ref{thm:mainVector} can be verified and our result describes the evolution of $\psi_\epsilon(s)$ that started at $s=0$ in the ground state space. 

Somewhat less generally, we also find interesting examples of the form~(\ref{Perturbation}) where the isolated spectral patch does not lie at the bottom of the spectrum. In~\cite{yarotsky2004perturbations}, where $H_0$ describes independent spins, and in~\cite{Bravyi:2011ea}, where $H_0$ can be taken to be the toric code model, it is proved that for $\vert \alpha \vert$ small enough the spectrum $\Sigma_s$ of $H_s$ is of the form
\begin{equation*}
\Sigma_s\subset  \bigcup_{ n \in \bbN }  B_{n},\qquad    B_{n}= \{ z: | z-E_n | \leq (C_0 + C_1 n) {\alpha } \},
\end{equation*}
where $E_n$ are the eigenvalues of $H_0$ and the constants $C_0,C_1>0$ are independent of the volume. If $\vert \alpha \vert$ is small enough, the low-lying bands $B_{0},B_{1},\ldots, B_{k}$ contain separated patches of spectrum. In particular, the band $B_{s,0}$ will contain the branches of eigenvalues arising from the possibly degenerate ground states of $H_0$. Furthermore, one may suspect that the $j$th band can be related to an effective $j$-particle subspace, see~\cite{yarotsky2004quasi} for weakly interacting spins and ~\cite{bachmann2016lieb} for a scattering picture. In this setting, Theorem~\ref{thm:mainVector} shows that the adiabatic evolution takes place mostly within the bands, with vanishing leaks between different bands --- as tested by local observables.

\subsubsection{Adiabatic dynamics within gapped ground state phases}

A gapped ground state phase of a quantum spin system is usually understood as a set of Hamiltonians defined on the same spin system, such that they all have a spectral gap above the ground state energy, and there is a piecewise $C^1$-path of gapped Hamiltonians interpolating between any pair, see for example~\cite{Sachdev:2000aa, Wen, automorphic, BachmannOgata}. It follows from the quasi-adiabatic flow technique described below that their ground states are locally unitarily equivalent, which is a structural result about the gapped phases. 

The adiabatic theorem proved here adds a dynamical aspect to the equivalence of states within gapped phases, in that it ensures that if a dynamics is started in the ground state of the initial Hamiltonian and the system is slowly driven along a gapped path to a Hamiltonian within the same phase, then the final state is a ground state of the final Hamiltonian, up to small diabatic errors. As a concrete example, this implies by~\cite{PVBS} that the coupling constants of the AKLT model, which is believed to belong to the phase of the antiferromagnetic spin-$1$ Heisenberg model, can slowly be changed so as to reach a product state with vanishing errors. Moreover, how slowly the process must be run to remain within a given error bound is independent of the length of the spin chain.

\subsection{The adiabatic evolution of the projector $P_0$}

Theorem~\ref{thm:mainVector} states that the driven Schr\"odinger evolution of a vector initially in the spectral patch  $\Sigma^1_0$ evolves within the spectral patch  $\Sigma^1_s$, up to small diabatic errors. Not surprisingly, the method of proof can be adapted to deal with the spectral projector associated with the complete patch. This version will be particularly suited to discuss the relation of our theorem to the `quasi-adiabatic evolution' results we discuss below.

For this, we consider the {von Neumann} equation for non-negative trace-class operators
\begin{equation}\label{SchDensity}
\iu \epsilon \frac{d}{ds} \rho_\epsilon(s) = [H_s,\rho_\epsilon(s)].
\end{equation}
The flow corresponding to~(\ref{SchDensity}) preserves positivity, preserves the trace, and if the initial operator is a projection, then so is the solution of the equation for all $s$. We shall consider the distinguished initial condition given by the spectral projection $P_0$, and denote $P_\epsilon(s)$ the solution of the equation. In other words, 
\begin{equation*}
P_\epsilon(s) = P_\epsilon(s)\str = P_\epsilon(s)^2
\end{equation*}
is the solution of
\begin{equation}\label{Schroedingerdensity}
\iu \epsilon \frac{d}{ds} P_\epsilon(s) = [H_s,P_\epsilon(s)],\qquad P_\epsilon(0) = P_0.
\end{equation}
The adiabatic theorem then suggests that $P_\epsilon(s)$ remains close to $P_s$ for $s\in[0,1]$, and indeed the following holds: Under the conditions of Theorem~\ref{thm:mainVector}(i), 
\begin{equation}\label{Adiabatic patch}
\sup_{s\in[0,1]}\left\vert \frac{\Tr (P_\epsilon(s) O)}{\Tr (P_\epsilon(s))} - \frac{\Tr (P_s O)}{\Tr (P_s)}\right\vert \leq C |\mathrm{supp}(O)|^2  \Vert O\Vert \epsilon
\end{equation}
uniformly in the volume $\Lambda$. This shows that the initial projector $P_0$ is parallel transported to $P_s$ by the adiabatic evolution, up to small diabatic errors. Note also that similarly to Theorem~\ref{thm:mainVector}(ii), the error becomes smaller than a higher power in $\epsilon$ once the driving has stopped, the power being bounded only by the degree of smoothness of the Hamiltonian.

\subsection{Connection to the quasi-adiabatic flow}\label{QAFlow}
Paraphrasing the above, one could view Theorem~\ref{thm:mainVector} as the statement that the flow $s\mapsto P_\epsilon(s)$ is closely related to the map $s\mapsto P_s$. They are indeed intimately connected and, as we shall see in the proofs, understanding the latter is crucial to understanding the former, but they should not be mistaken for one another.

A central element in the proof of the theorem is the fact that $s\mapsto P_s$ can be constructed as a flow that is generated by a local Hamiltonian --- although not, of course, the Hamiltonian $H_s$. Precisely, there exists a $K_s \in \caB_{\caS,\infty}$ such that
\be \label{eq: qa flow}
\dot P_s= \iu [K_s,P_s],
\ee
see Corollary~\ref{cor:quasi-ad} below which provides a shorter proof than { the original proof in}~\cite{automorphic,HastingsWen}. While $K_s$ is often referred to in the literature as the generator of the `quasi-adiabatic evolution' or even `adiabatic evolution', we shall call it here the \emph{generator of the spectral flow} to avoid any confusion.

\subsection{Comparison with the traditional quantum adiabatic theory}
In the usual adiabatic theory \cite{Teufel, AFGG}, the solution of the von Neumann equation is expanded in the powers of $\epsilon$, 
$$
\rho_\epsilon(s) = P_s + \sum_{j=1}^n \epsilon^j a_{j}(s)  + \epsilon^{n}R_{n, \epsilon}(s).
$$
With appropriate regularity assumptions and initial conditions, operators $a_{j}(s)$ can be determined from a recurrence relation, the coefficient $a_{j+1}(s)$ being a linear function of $a_{j}(s)$ and $\dot{a}_{j}(s)$. The remainder term $R_{n, \epsilon}(s)$ is then obtained by Duhamel's formula,  
$$
R_{n, \epsilon}(s) = -\int_0^s\sigma_{s,s'}(r_{n}(s'))\dd s', \quad r_{n}(s) :=\frac{d}{ds}\left[(1-\mathcal{P}_s)a_{n}(s)\right],
$$
where $\sigma_{s,s'}(\cdot)$ is the flow generated by~\eqref{SchDensity} and $\mathcal{P}_s$ is the ergodic projection of $L_s=-\iu[H_s,\cdot]$.  

If $a_{j}, r_{n}$ were local operators in some sense, Theorem~\ref{thm:mainVector} could be established from such an expansion. Let us consider consider the first term of the expansion, given by $a_1(s) = \iu L_s^{-1} ([K_s,P_s])$. As discussed in Proposition~\ref{prop:HInverse}, one can invert $L_s$ so that $a_1(s) = \iu [\Upsilon(s), P_s]$, where $\Upsilon(s)$ is a local Hamiltonian. Then we have 
$$
\Tr(O a_1(s)) = -\iu\Tr(P_s [\Upsilon(s), O]).
$$
for any local observable $O$. Even though the norm of $a_1(s)$ grows as the volume, the expression in the trace manifestly does not because $\Upsilon(s)$ enters only in the commutator with the observable.

One might hope --- as the authors initially did --- that it is possible to proceed with $a_2, \ldots, a_n$ and $r_n$ in a similar manner. Unfortunately, we were not able to do so. Even if all the terms could be expressed as nested commutators so that traces involving $a_1,\ldots,a_n$ would be bounded uniformly in the volume, the remainder term $R_{n,\varepsilon}$ would diverge as $\epsilon\to 0$, since $\sigma_{s,s'}$ spreads the support of the local observable over a distance $\epsilon^{-1}$.

{We were only able to circumvent this issue} by introducing a new way of expanding the solution which carefully respects the locality of the dynamics. This is explained in the following section.

\subsection{Main idea of the proof} \label{rem:Proof}
We construct a sequence of local Hamiltonians $\{A_\alpha:1\leq \alpha\leq n\}$ generating a sequence of \emph{local unitary dressing transformations} of order $n$
\begin{equation}\label{ExpAnsatz}
U_{n,\epsilon}(s):= \exp\left(\iu S_{n,\epsilon}(s)\right),\qquad S_{n,\epsilon}(s) = \sum_{\alpha=1}^n\epsilon^\alpha A_\alpha(s),
\end{equation}
which closely follows the Schr\"odinger propagator $U_\epsilon(s,0)$ in the sense that
\begin{equation*}
\iu\epsilon \frac{d}{ds}U_{n,\epsilon}(s) = (H_s + R_{n,\epsilon}(s))U_{n,\epsilon}(s),
\end{equation*}
where the \emph{counter-diabatic driving} $R_{n,\epsilon}(s)$ is a local Hamiltonian of order $\epsilon^{n+1}$. The local Hamiltonians $A_\alpha(s)$ are determined recursively, however, unlike the standard expansion, the operator $A_\alpha(s)$ is a polynomial function of all the previous operators and their derivatives.

Like in the wishful argument of the previous section, the difference between the dressed projection $\Pi_{n,\epsilon}(s) := U_{n,\epsilon}(s) P_sU_{n,\epsilon}(s)\str$ and the solution of the Schr\"odinger equation $\rho_\epsilon(s)$ is locally of order $\epsilon^{n+1-d}$. By construction $S_{n,\epsilon}(s)$ is of order $\epsilon$, so that the difference between the dressed ground state and the ground state itself is of order $\epsilon$. Hence, the theorem follows if $n$ can be chosen larger than $d$, which can be done if the Hamiltonian is smooth enough.

Furthermore, the Hamiltonians $A_\alpha(s)$ depending locally-in-time on $H_s$ and its derivatives, they vanish whenever the driving stops. At that point $S_{n,\epsilon}(s)$, which is generically of order $\epsilon$, in fact vanishes and the dressing transformation becomes trivial. This allows for the improved bound $\epsilon^m$ at times $s$ where the Hamiltonian has become again time-independent.


\section{A corollary: Linear response theory}\label{sec:LRT}

In this section we present the important theoretical application of Theorem~\ref{thm:mainVector} that was announced in the introduction, namely that the theorem allows for a proof of the validity of linear response theory in the case of gapped systems.

To set the stage, we introduce the key actors in the particular case of the integer quantum Hall effect. We consider a quantum Hall sample of an area $L^2$ in an external magnetic field $B$. When the density of electrons is $n = B/(2 \pi q)$ for some $q \in \mathbb{Z}$, the non-interacting system is gapped and the gap remains open also in the presence of interactions \cite{Giuliani:2016gn}. If such a system is in its ground state and the Fermi energy lies in the gap there is no net motion of electrons across the sample. When an electric field $E_\nu$ is applied, it results in a current $j_\nu$ in the perpendicular direction. It is observed that for weak driving the current density is proportional to the applied force, $j_\mu =  f_{\mu \nu} E_\nu$, where $f_{\mu \nu}$ are the so-called response coefficients, in this particular case the conductances. In the quantum Hall effect the matrix of response coefficients is as discussed off-diagonal and these non-zero elements are equal to $1/(2 \pi  q)$ for integer $q$. Beyond the simple linear relation between driving and current, linear response theory provides an explicit formula for the matrix $f$, first derived by Kubo~\cite{Kubo:1957cl}. The quantisation of the Hall conductance as defined by Kubo's formula has been established in various degrees of generality, and in particular including interactions \cite{HastingsMichalakis,Giuliani:2016gn}. However, the validity of the formula itself has not yet been proved from first principles, in particular not in a many-body situation. The formula was established for non-interacting Landau type Hamiltonians in \cite{Elgart:2004wd}, and in the presence of disorder in~\cite{LRLocal}.
Progress was also made for the response smoothed in frequency-domain, see e.g.~\cite{klein2007mott,Bru:2017aa} and the references therein, and for a strictly local driving in a thermal setting, see and \cite{AbouSalem:2005kr,Jaksic:2006fv}.

We adapt Kubo's framework of linear response theory to the setting of quantum spin systems. Let $H_{\mathrm{initial}}$ be the unperturbed Hamiltonian in the infinite past upon which a perturbation $V$ is slowly switched on to reach an order $\alpha$ at $t=0$. The full driven Hamiltonian has the form
\begin{equation*}
H_{\epsilon t,\alpha} = H_{\mathrm{initial}} + \ep{\epsilon t} \alpha V,\qquad t\in(-\infty,0].
\end{equation*}
The response coefficient is associated to a local observable $J$, typically a current. Let $P_\alpha$ be the projection on the ground state space of $H_{0,\alpha}$. Note that $P_0$ is the projection onto the ground state space of $H_{\mathrm{initial}}$. As in the previous section, $P_{\epsilon,\alpha}(t)$ shall denote the solution of the driven Schr\"{o}dinger equation generated by $H_{\epsilon t,\alpha}$ with an initial condition $P_{\epsilon,\alpha}(-\infty) = P_{0}$. In finite volume, these projections can be interpreted as density matrices. Choosing the perspective of states as functionals on observables, we define
$$
\omega_{\epsilon,\alpha;t}(O)= \frac{\Tr (O P_{\epsilon,\alpha}(t))}{\Tr P_{\epsilon,\alpha}(t)}, \qquad \omega_{\alpha}(O)= \frac{\Tr (O P_{\alpha})}{\Tr P_{\alpha}},
$$
for any local observable $O$, keeping in mind that all these objects are volume-dependent. Note that the denominators are constant in the parameters $\epsilon,\alpha,t$, but they generically change with the volume.

Formally, the response coefficient $f_{J,V} $ is then given as the linear in $\alpha$ response to the perturbation~$V$, in the adiabatic limit: 
 \begin{equation}
\label{LRdef}
\omega_{\epsilon,\alpha;0}(J) - \omega_{\alpha}(J) -\alpha f_{J,V} =  o(\alpha), \qquad \text{as $\alpha\to 0, \epsilon \to 0$}.
\end{equation}
What makes the claim of validity of linear response truly non-trivial is that the limit $\epsilon\to 0$ is taken first and that the bound $o(\alpha)$ is uniform in the volume. In particular, if $o(\alpha)$ were allowed to depend on the volume $\Lambda$, then the resulting claim may physically be meaningless, see the arguments in~\cite{Kampen:1971aa}.

Let us now see how the above setup connects to adiabatic theory. Since the time-dependence is slow in the $\epsilon\to 0$ limit, it is justified to approximate the evolved projection $P_{\epsilon,\alpha}(0)$ by the instantaneous ground state projection $P_{\alpha}$, and hence the evolved state $\omega_{\epsilon,\alpha;0}$ by the instantaneous $\omega_{\alpha}$. Moreover, the replacement can be done uniformly in the volume, since the observable $J$ is a fixed local observable. In order to obtain an expression for the response coefficient  $f_{J,V}$, we then simply have to expand $\omega_{\alpha} (J)$ in $\alpha$. As long as the gap remains open for the Hamiltonians, the first order term is obtained using the generator of the spectral flow $K_\alpha$, see~(\ref{eq: qa flow}). Using cyclicity of the trace at $\alpha=0$, we obtain 
$$
\omega_\alpha(J)=\omega_0(J) - \iu\alpha\omega_{0}\left( [K_{0},J]\right)+o(\alpha),
$$
thereby identifying $f_{J,V}= -\iu\omega_{0}\left( [K_{0},J]\right)$. The precise theorem we prove is
\begin{thm}\label{thm:LR}
Suppose that Assumptions~\ref{A:patch}, \ref{A:Hamiltonians} hold for the family of Hamiltonians $H_{\mathrm{initial}} + \sigma \alpha  V$, for some fixed $\alpha>0$ and $\sigma\in[0,1]$. Let $J \in \mathcal{A}^X$, with $X \subset \Lambda$, be  an observable, and let $f_{J,V}= -\iu\omega_{0}\left( [K_{0},J]\right)$. 
 Then the expression
\begin{equation}
\label{Kubo}
 \alpha^{-1}\left(\omega_{\epsilon,\alpha;0}(J) - \omega_{0}(J) -\alpha f_{J,V}\right)
\end{equation}
converges to $0$, uniformly in the volume $\Lambda$, as first $\epsilon\to 0$ and then $\alpha\to 0$.
\end{thm}
In fact, the convergence works for any coupled limit $(\alpha,\epsilon)\to (0,0)$ but we stressed the physical order of limits.
The expression for $f_{J,V}$ given above might not be familiar. A more recognizable expression that often appears in expositions of linear response is rather 
\begin{equation}\label{eq: standard expression kubo}
f_{J,V}= \iu \lim_{\delta \downarrow 0} \int_{0}^{\infty} dt \, e^{-\delta t}\omega_{0}\left( [\tau_{-t}(V),J]\right),
\end{equation} 
where $\tau_t(V)=e^{\iu t H_{\mathrm{initial}}} V e^{-\iu t H_{\mathrm{initial}}}$. We show in Section \ref{sec: proof of main theorem} that within our setup, this expression indeed coincides with $f_{J,V}$ as given in Theorem \ref{thm:LR}. 

Though it goes beyond the framework of this paper, it is worthwile to hint at an infinite-volume formulation. In this formalism, one considers observables $O$ in the \emph{quasilocal} algebra $\caA$, defined as the norm-closure of the inductive limit of $\caA^{\Lambda},\Lambda\nearrow \Gamma$. The set of states $\caS(\caA)$ consists then of positive, continuous  linear functionals on $\caA$, see e.g.~\cite{BratteliRobinsonBook} for precise definitions.   
Then, in the weak-* topology, the family of ground states $\{\omega_\alpha^\Lambda:\Lambda\in\caF(\Gamma)\}$ has accumulation points in $\caS(\caA)$.  Let us assume for simplicity that there is unique limit, namely that there is a state $\overline{\omega}_\alpha$ such that
\begin{equation*}
\lim_{\Lambda\nearrow \Gamma} \omega^{\Lambda}_\alpha(O) = \overline{\omega}_\alpha(O)
\end{equation*}
for all local $O$. In that case, it follows that the dynamically-defined states $\omega_{\epsilon,\alpha;t}$ have thermodynamic limits $\overline{\omega}_{\epsilon,\alpha;t}$ in the sense above, and the theorem can be reformulated as
\begin{equation}
\label{Kubo again}
\lim_{\alpha\to 0}\lim_{\epsilon\to 0}\alpha^{-1}\left(\overline{\omega}_{\epsilon,\alpha;0}(J) - \overline{\omega}_{0}(J)\right)
 =  -\iu\overline{\omega}_{0}\left( [K_{0},J]\right),
\end{equation}
where $[K_{0},J] = \lim_{\Lambda\to\Gamma}[K_{0}^\Lambda,J]$ with the limit meant in $\caA$.

Another extension can be easily obtained from the theorem in a translation-invariant setting, where the Hamiltonians and the states are all translation invariant. Then the theorem extends to the current density $j = \vert \Lambda \vert^{-1} J$ of an extensive current $J$ which is obtained as the sum of translates of the same local operator. It is in particular instructive to consider the response of the initial energy $H_{\mathrm{initial}}$ in this setting. The associated response coefficient $-i\omega_0([K_0, H_{\mathrm{initial}}])$ is zero in view of $[H_{\mathrm{initial}}, P_0]=0$. This is a manifestation of the fact that response of a gapped quantum system is always non-dissipative.


\section{Proofs}\label{sec:Proofs}

\subsection{Setup for locality} \label{sec: infinite volume setup}

As pointed out in the previous sections, the locality of the dynamics of a quantum spin system will play a key role in the proofs. Here, we detail the setup describing the locality properties of the lattice and of the Hamiltonians, and recall the central estimate in Section~\ref{sub:LRB}: the Lieb-Robinson bound. We shall further show that the set of local Hamiltonians is closed under two natural operations: taking a commutator and the inverse thereof. We treat these operations in Sections~\ref{sec: commutators} and~\ref{sec: the map i}.

\subsubsection{An integrable structure on $\Gamma$}
To the countable graph $\Gamma$, we associate the  function $F(r)=(1+r)^{-(d+1)} $ where $d$ is the dimensionality of $\Gamma$, as defined by \eqref{ddim}. This function will be assumed to be fixed throughout and we mostly omit it from the notation.   
Its important properties \cite{Nachtergaele:2006bh} are the existence of a constant $C_F<\infty$ such that
\begin{equation*}
\sum_{z\in\Gamma}F(d(x,z))F(d(z,y))\leq C_FF(d(x,y))
\end{equation*}
for all $x,y\in\Gamma$, and that
\begin{equation*}
\Vert F\Vert_1:= \sup_{x\in\Gamma} \sum_{z\in\Gamma}F(d(x,z))<\infty.
\end{equation*}
For any bounded, non-increasing, positive function $\zeta:[0,\infty)\to(0,\infty)$ that is logarithmically superadditive, namely $\zeta(r+s)\geq\zeta(r)\zeta(s)$, we let $F_\zeta(r):= \zeta(r) F(r)$. Then, $F_\zeta$ has the two properties above with $C_{F_\zeta}\leq C_F$.

\subsubsection{Interactions}
We already introduced the notion of an \emph{interaction} as a family of self-adjoint operators $\{\Phi(Z):Z \in \caF(\Gamma)\}$, possibly time-dependent $\Phi(Z)=\Phi_t(Z)$, stressing that such interactions allow to define a family of Hamiltonians $\{H^\Lambda: \Lambda\in \caF(\Gamma)\}$ and therefore of dynamics, indexed by finite volumes.

For the proofs below, it is natural to consider a slightly more general setup, namely that of a family of $\Lambda$-dependent interactions $\{\Phi^\Lambda:\Lambda\in\caF(\Gamma)\}$, such that $\Phi^\Lambda(X)=0$ unless $X \subset\Lambda$. These interactions are not necessarily compatible, i.e.\ possibly
$\Phi^\Lambda(X)\neq \Phi^{\Lambda'}(X)$, even if both sides are nonzero. We denote these families of interactions simply by $\Phi$ and we call this object an interaction as well. Some compatibility requirement could be added, for example to ensure the existence of a well-defined infinite volume limit of the associated dynamics, see e.g.\ Section~5 of~\cite{automorphic}, but this will not be necessary here.

 For $n\in\bbZ$ and $\zeta$ as above, we define a norm on (possibly time-dependent) interactions $\Phi$:
\begin{equation*}
\Vert \Phi\Vert_{\zeta,n}:= \sup_{\Lambda\in\caF(\Gamma)} \sup_{x,y\in\Gamma}\sum_{Z\ni\{x,y\}}\sup_{t\in\bbR}\vert Z\vert^n\frac{\Vert \Phi_t^\Lambda(Z)\Vert}{F_\zeta(d(x,y))}.
\end{equation*}
We let $\caB_{\zeta,n}$ stand for the Banach space of interactions completed in this norm, which depends on the function $\zeta$ and the power $n$ quantifying the locality. Note that $\caB_{\zeta,n}\subset\caB_{\zeta,m}$
whenever $n> m$.

We shall encounter two particular classes of functions $\zeta$ in the sequel. The first one are the exponentials, $\zeta(r) = \exp(-\mu r)$ for a $\mu>0$.  The second class is that of functions decreasing faster than any inverse power,
\begin{align*}
\caS := \big\{\zeta:[0,\infty)\to(0,\infty) : \:&\zeta \text{ is bounded, non-increasing, logarithmically superadditive and }\\ &\sup\{ r^n \zeta(r): r\in[0,\infty)\}<\infty\text{ for all }n\in\bbN\big\}.
\end{align*}
In the following, the exact decay rates of the exponential or of a function in $\caS$ will be irrelevant, while we will often need the power $n$ to be arbitrary large. Therefore, we define the classes of interactions
\begin{equation*}
\caB_{\caE,n}= \cup_{\mu>0}  \caB_{\exp(-\mu\cdot),n},\qquad \caB_{\caS,n} = \cup_{\zeta \in \caS}  \caB_{\zeta,n},
\end{equation*}
(note that $\caB_{\caE,n}\subset\caB_{\caS,n}$), and
\begin{equation*}
\caB_{\caE,\infty} = \cap_{n \in \bbN} \caB_{\caE,n},\qquad \caB_{\caS,\infty} = \cap_{n \in \bbN} \caB_{\caS,n},
\end{equation*}
which are the classes appearing in the main theorem.

\subsubsection{Local Hamiltonians}
An interaction allows us to build, for any finite volume $\Lambda$ and time $t$, an associated Hamiltonian by 
\begin{equation}\label{eq:HDecomp}
H^\Lambda_{t}=\sum_{Z \in \caF(\Gamma)} \Phi_t^\Lambda(Z)= \sum_{Z \subset \Lambda} \Phi_t^\Lambda(Z).
\end{equation}
We will often turn the logic around and treat the Hamiltonians as the central object. In particular, we shall say that a family of operators $\{H^\Lambda:\Lambda\in\caF(\Gamma)\}$ is a \emph{local Hamiltonian} if there exists an interaction $\Phi_H$ such that \eqref{eq:HDecomp} holds. We shall call $\Phi_H$ an `interaction associated to $H$', and we note that $\Phi_H$ is not unique because the decomposition~(\ref{eq:HDecomp}) is not unique. 
We denote by $\caL_{\caE,n}, \caL_{\caS,n}, \caL_{\caE,\infty}, \caL_{\caS,\infty}$, the set of local Hamiltonians $H$ for which there is an interaction $\Phi_H$ in $\caB_{\caE,n}, \caB_{\caS,n}, \caB_{\caE,\infty}, \caB_{\caS,\infty}$, respectively.

Finally, we shall denote $H\in\caC^k$ if the $j$th time derivative $H^{(j)}$ of a Hamiltonian (as a matrix-valued function) exists for all $1\leq j\leq k$. Note that even if $H$ belongs to one of the `nice classes' $\caL$, its derivatives do not necessarily do so. For the particular Hamiltonian defining the dynamics in this paper, Assumption~\ref{A:Hamiltonians} however precisely ensures that $H^{(j)}\in\caL_{\caE,\infty}$ for all $1\leq j\leq d+m$.

\subsubsection{Lieb-Robinson bounds}\label{sub:LRB}

A local Hamiltonian generates a dynamics $\tau^\Lambda_{t,t_0}$ on $\caA^\Lambda$, i.e.\ a cocycle of automorphisms $\caA^\Lambda\mapsto \caA^\Lambda$ by 
$$
-\iu \frac{d}{dt}\tau^\Lambda_{t,t_0}(O) = [H^\Lambda_{t}, \tau^\Lambda_{t,t_0}(O) ], \qquad   \tau^\Lambda_{t_0,t_0}(O)=O.
$$
This dynamics satisfies a Lieb-Robinson bound \cite{Lieb:1972ts,Nachtergaele:2006bh}: If $\Phi_H\in\caB_{\zeta,n}$ for some $n\in\bbN\cup\{0\}$ and $\zeta\in\caS$, there is a constant $C_{\zeta}$ such that
\begin{equation}\label{LR bound}
\left\Vert [\tau^\Lambda_{t,t_0}(O^X),O^Y]\right\Vert\leq \frac{2\Vert O^X\Vert \Vert O^Y\Vert}{C_{\zeta}} \ep{2 C_{\zeta} \Vert \Phi \Vert_{\zeta,0}\vert t-t_0 \vert}\sum_{x\in X,y\in Y} F_\zeta(d(x,y))
\end{equation}
for all $O^X\in\caA^X, O^Y\in\caA^Y$, with $X,Y \subset \Lambda$. If $d(X,Y)>0$, then
\begin{equation*}
\sum_{x\in X,y\in Y} F_\zeta(d(x,y)) \leq \Vert F \Vert_1 \min\{\vert X\vert ,\vert Y\vert\}\zeta(d(X,Y)),
\end{equation*}
so that $\zeta$ expresses the decay of the commutator for times that are short compared to $d(X,Y)$. In particular, in the case $\zeta(r) = \exp(-\mu r)$, we define the Lieb-Robinson velocity
\begin{equation}\label{LRv}
v := \frac{2 C_{\zeta} \Vert \Phi \Vert_{\zeta,0}}{\mu}.
\end{equation}

\subsubsection{Commutators} \label{sec: commutators}
 Let $H,G\in\caL_{\caS,n}$ be two local Hamiltonians, with associated interactions $\Phi_H,\Phi_G$. We define their commutator as the local Hamiltonian $J$ with
\begin{equation*}
J^\Lambda= [H^\Lambda, G^\Lambda],\qquad \Lambda\in\caF(\Gamma),
\end{equation*}
and we check that $J \in\caL_{\caS,n-1}$. Indeed, a family of interactions $\Phi_J^\Lambda$ for $J$ is given by
\begin{equation}\label{Interaction Commutator}
\Phi_J^\Lambda(Z) = \sum_{\substack{X,Y\subset\Lambda: \\ X\cup Y = Z, X\cap Y\neq \emptyset}}[\Phi_H^\Lambda(X), \Phi_G^\Lambda(Y)],
\end{equation}
and Lemma~\ref{lem:IntComm}(ii) shows that indeed $\Vert \Phi_J \Vert_{\caS,n-1}<\infty$.

\subsubsection{The map $\caI_s$} \label{sec: the map i}

We now introduce a tool which plays an central role in the proof of the main theorem. For this, we consider only Hamiltonians
$H_s \in \caL_{\caE,\infty}$ featuring in our main Theorem \ref{thm:mainVector}: they are associated with an interaction decaying exponentially and are gapped with a gap $\gamma$, see Assumption~\ref{A:patch}. First, let $\tau^{s,\Lambda}_t$ be the dynamics generated by this local Hamiltonian, but with $s$ frozen, i.e.\ we consider here a time-independent Hamiltonian and therefore $\tau^{s,\Lambda}_t$ carries only one time-subscript. It is a group and not merely a cocycle. Second, let $W_\gamma\in L^1(\bbR)$ be a function such that $\sup\{\vert t \vert^{n} \vert W_\gamma(t)\vert:\vert t \vert >1\}<\infty$ for all $n\in\bbN$, and such that its Fourier transform satisfies
\begin{equation*}
\widehat W_\gamma(\zeta) = \frac{-\iu}{\sqrt{2\pi}\zeta},\quad\text{if}\quad \vert \zeta\vert \geq \gamma.
\end{equation*}
For an example of such a function, see \cite{automorphic}.

With this, a map $\caI^\Lambda_s: \caA^\Lambda\to\caA^\Lambda$ is defined by
\begin{equation}\label{I}
 \caI^\Lambda_s (A) := \int_\bbR W_\gamma(t)\tau^{s,\Lambda}_t(A)  \dd t.
\end{equation}
Now, we choose $A=G^\Lambda$, the finite-volume version of a local Hamiltonian $G\in\caL_{\caS,n}$ with associated interaction $\Phi_G$, and we define $\caI_s$ to be map on local Hamiltonians:
\begin{equation*}
\caI_s(G) := \{\caI_s^\Lambda(G^\Lambda) : \Lambda\in\caF(\Gamma))\},
\end{equation*}
Lemma~\ref{lem:IofPhi}(i) shows that $\caI_s(G) \in\caL_{\caS,n-1}$.

As an application of this construction, one can choose $G=\dot H_s$, with $H_s$ the same Hamiltonian at frozen time $s$ as the one generating the dynamics $\tau^{s,\Lambda}_t$. Then 
\begin{equation*}
K_s = \caI_s(\dot H_s),
\end{equation*}
is the generator of the spectral flow~\cite{Hastings:2004go,HastingsWen,Osborne,automorphic} already mentioned, see Corollary~\ref{cor:quasi-ad}.


\subsubsection{Conventions}
In what follows, we drop as much as possible the superscripts $\Lambda$ on $H^\Lambda$ and $\Phi^\Lambda$.  We do this to make the notation lighter and because we believe that confusion is mostly excluded.  As a matter of fact, one could ignore the setup of families of Hamiltonians and interactions completely and simply imagine that we do all of our analysis in a fixed, large volume. With this picture in mind, the upshot of the results is that all bounds do not depend on the volume. 

Furthermore, from now on the notation $H$ will no longer refer to a general local Hamiltonian, but to the specific Hamiltonian that features in our main Theorem~\ref{thm:mainVector} and that satisfies Assumptions~\ref{A:patch} and~\ref{A:Hamiltonians}. 

Finally, we will also drop the subscript $s$ since all objects but for the `test observables' $O$ depend on $s$. As already used above, we usually denote $\dot A = \frac{d}{ds}A$.

\subsection{Proof of the main theorem} \label{sec: proof of main theorem}

We first gather some remarkable properties of the map $\caI$, which can be summarised as follows: under suitable conditions, $\caI(\cdot)$ provides a local inverse of $[H,\cdot]$. This will be repeatedly used in the proof of Theorem~\ref{thm:mainVector}.
\begin{prop}\label{prop:HInverse}
Let $H$ be a local Hamiltonian satisfying Assumptions~\ref{A:patch} and~\ref{A:Hamiltonians}. Recall that $P$ denotes the spectral projection onto the spectral patch $\Sigma^1$.
\begin{enumerate}
\item If an operator $A$ satisfies the \emph{off-diagonal condition}
\be  \label{eq: offdiagonal condition}
A = P A(1-P)+ (1-P) AP,   
\ee
then 
\begin{equation*}
A = -\iu[H, \caI(A) ].
\end{equation*}
\item For any operator $L$,
\begin{equation*}
[L,P] -\iu [[\caI(L),H],P] =0.
\end{equation*}
\item If $H\in\caL_{\caE,\infty}$ and $G\in\caL_{\caS,\infty}$, then $\caI(G)\in\caL_{\caS,\infty}$
\item If $H,G$ depend on a parameter with $H,G\in \caC^k$ and $H^{(j)}\in\caL_{\caE,\infty},G^{(j)}\in\caL_{\caS,\infty}$ for $0\leq j\leq k$, then $\caI(G) \in \caC^{k}$ with $\caI(G)^{(j)}\in\caL_{\caS,\infty}$ for $1\leq j\leq k$.
\end{enumerate}
\end{prop}
\begin{proof}
(i) If $H = \int \lambda \dd E(\lambda)$ denotes the spectral decomposition of $H$, we have that
\begin{equation*}
-\iu [H, \caI(A)] = \sqrt{2\pi} \int \iu\widehat{W}_\gamma(\mu-\lambda) (\mu-\lambda) \dd E(\lambda) A  \dd E(\mu).
\end{equation*}
which proves the claim since $A$ is off-diagonal and $\widehat W_\gamma(\xi)= \frac{-\iu}{\sqrt{2\pi}\xi}$ for $ | \xi | \geq \gamma$. \\ 
\noindent (ii) For any $L$, the operator $A= [L,P] $ is off-diagonal in the sense of \eqref{eq: offdiagonal condition}. Therefore, by~(i), 
$$
[L,P] +\iu [H, \caI([L,P])]=0.
$$
Since $[H,P]=0$, we have $\caI([L,P])=[\caI(L),P]$ and hence 
$$
[L,P] -\iu [P,[H,\caI(L)]]=0
$$
by the Jacobi identity and again $[H,P]=0$. \\
\noindent (iii) The fact that $ \caI(\Phi) \in\caL_{\caS,\infty}$ is a special case of Lemma~\ref{lem:IofPhi}(i) below. \\
\noindent (iv) We first note that
\begin{equation}\label{diff of I}
\frac{d}{ds}\caI(G) = \caI(\dot G) + \iu \int_\bbR {W_\gamma(t)\tau_t\left(\left[\int_0^t\tau_{-u}(\dot H)\dd u,G\right]\right)}\dd t,
\end{equation}
which proves that $\caI(G)\in\caC^{1}$. By Lemma~\ref{lem:IofPhi}(i) and the assumption, $ \caI(\dot G) \in\caL_{\caS,\infty}$. Moreover, the assumption and the Lieb-Robinson bound imply that $\int_0^1\tau_\alpha(\dot H)\dd\alpha\in\caL_{\caE,\infty}$, so that the commutator belongs to $\caL_{\caS,\infty}$ by Lemma~\ref{lem:IntComm}(iii). The argument of Lemma~\ref{lem:IofPhi}(i) applies and proves that the second term above belongs to $\caL_{\caS,\infty}$ as well. This proves (iv) for $j=1$.

Furthermore, $\caI(G)^{(j)}$ similarly depends on $H^{(j')},G^{(j')}$ for $0\leq j'\leq j$, and hence $\caI(G)\in\caC^k$ if $H,G\in\caC^k$. The fact that $\caI(G)^{(j)}\in\caL_{\caS,\infty}$ follows again from a combination of the Lieb-Robinson bound, Lemma~\ref{lem:IntComm} and the proof of Lemma~\ref{lem:IofPhi}(i).
\end{proof}
As mentioned above, this proposition provides an elementary proof of the fact that $K=\caI(\dot H)$ is a local generator of the spectral flow.
\begin{cor}\label{cor:quasi-ad}
If Assumptions~\ref{A:patch} and~\ref{A:Hamiltonians} hold, then
\begin{equation*}
\dot P = \iu [\caI(\dot H), P]\,.
\end{equation*}
\end{cor}
\begin{proof}
Since $\dot P$ is off-diagonal in the sense of \eqref{eq: offdiagonal condition}, Proposition~\ref{prop:HInverse}(ii) implies that
\begin{equation*}
\dot P = -\iu[H,\caI(\dot P)] = -\iu \caI([H,\dot P]) = \iu \caI([\dot H, P]) = \iu [\caI(\dot H), P],
\end{equation*}
where we used repeatedly that $[H,P]=0$.
\end{proof}
Although this does not play a role here, it is natural in the adiabatic setting to consider the block decomposition of the generator $K_s$ with respect to the projection $P_s$. In general, there is nothing particular to mention, and in particular, $(1-P_s)K_s(1-P_s)\neq 0$. However, if $P_s$ corresponds to an exactly degenerate eigenvalue, $P_s H_s P_s = E_s P_s$ for some $E_s\in\bbR$, then
\begin{equation*}
P_s K_s P_s = P_s \dot H_s P_s \int_{\bbR} W_\gamma(t) \dd t,
\end{equation*}
so that $K_s$ is trivial on $\mathrm{Ran}P_s$ if $W_\gamma$ is an odd function. In that particular case, the equation $\iu\frac{d}{ds}\Omega(s)=K_s\Omega(s),\Omega(0) = \psi_0$ has the same solution as the parallel transport equation~(\ref{eq: parallel transport 2}).

We can now turn to the heart of the proof, namely the argument sketched in Remark~\ref{rem:Proof}. In the following lemma, we construct the counter-diabatic driving and the local dressing transformation.
\begin{lemma}\label{ConstructPhi}
Let Assumptions~\ref{A:patch},\ref{A:Hamiltonians} hold for some $m\in\bbN$. For any $n\leq \tilde m= d+m$, there are $\{A_\alpha,1\leq \alpha\leq n\}$ with $A_\alpha\in\caL_{\caS,\infty}$ such that the projector
\begin{equation*}
\label{eq: ansatz pi}
\Pi_{n,\epsilon} :=  U_{n,\epsilon} P U_{n,\epsilon}\str,   \qquad \text{with} \quad U_{n,\epsilon}=  \exp\left(\iu \sum_{\alpha=1}^n\epsilon^\alpha A_\alpha\right)
\end{equation*}
 solves 
\begin{equation}\label{PiDynamics}
\iu \epsilon \dot \Pi_{n,\epsilon} = [H + R_{n,\epsilon},\Pi_{n,\epsilon}]
\end{equation}
where $R_{n,\epsilon}\in\caL_{\caS,\infty}$ with associated potential $\Phi_{n,\epsilon}$ satisfying $\Vert\Phi_{n,\epsilon} \Vert_{\caS,k} \leq r_{n,k}(s)\epsilon^{n+1}$ for all $k\in\bbN$, where $r_{n,k}(s)$ is independent of $\epsilon$. \\ 
Moreover, $A_\alpha\in\caC^{(\tilde m-\alpha)}$, with $A^{(j)}_\alpha\in \caL_{\caS,\infty}$ for all $1\leq j\leq \tilde m-\alpha$.

\end{lemma}
\begin{proof}
For notational clarity, we drop all indices $\epsilon$ in the proof. The Ansatz $\Pi_n :=  U_n P U_n\str$ and~(\ref{eq: qa flow}) yield
\begin{align}\label{Pi equation}
\iu\epsilon \dot\Pi_n &= \iu\epsilon \dot U_n P U_n\str + \iu\epsilon  U_n P \dot U_n\str - \epsilon  U_n [K, P] U_n\str \\
&= [H,\Pi_n] + \left[\iu\epsilon \dot U_n U_n\str - \epsilon U_n K U_n\str + (U_n H U_n\str - H),\Pi_n\right]\nonumber
\end{align}
where we have used that $U_n\dot U_n\str = -\dot U_n U\str$ by unitarity and $[U_n H U_n\str, \Pi_n] = U_n [H,P] U_n\str = 0$. We write the second commutator as
\begin{equation*}
\left[\iu\epsilon \dot U_n U_n\str - \epsilon U_n K U_n\str + (U_n H U_n\str - H),\Pi_n\right] 
= U_n \left[\iu\epsilon  U_n\str \dot U_n - \epsilon K + H - U_n\str H U_n ,P \right] U_n\str
\end{equation*}
and aim for an expansion to finite order in $\epsilon$. For this, let $S_n = \sum_{\alpha=1}^n\epsilon^\alpha A_\alpha$ and
\begin{equation}\label{HExpansion}
U_n\str H U_n = \ep{-\iu \ad{S_n}}(H) = \sum_{k=0}^n \frac{(-\mathrm{i})^k}{k!} \Big(\sum_{\alpha=1}^n \epsilon^\alpha \ad{A_\alpha}\Big)^k (H) + \caO(\epsilon^{n+1})
\end{equation}
and define $U_n\str H U_n = : \sum_{\alpha=0}^n \epsilon^\alpha H_\alpha + \epsilon^{n+1}h_{n}(\epsilon)$, namely
\begin{equation}\label{Ha}
H_\alpha = \sum_{\mathbf j: s(\mathbf j) = \alpha}\frac{(-\mathrm{i})^k}{k!}\ad{A_{j_k}}\cdots\ad{A_{j_1}}(H),
\end{equation}
where the sum is over finite vectors $\mathbf{j}= (j_1,\ldots,j_k)$, and $s( \mathbf{j})= j_1+\ldots+j_k$. For the first few orders, this reads
\begin{equation*}
H_0 = H,\qquad H_1 = -\iu [A_1,H],\qquad H_2 = -\iu [A_2,H] - \frac{1}{2}[A_1,[A_1,H]].
\end{equation*}
Similarly, Duhamel's formula
\begin{equation*}
\mathrm{i} U_n^*\dot U_n = -\int_0^1 \ep{-\iu\lambda S_n} \dot S_n \ep{\iu\lambda S_n} \dd \lambda
\end{equation*}
can be expanded as
\begin{equation*}
\mathrm{i} U_n^*\dot U_n = \sum_{j=1}^n \epsilon^{j} \sum_{k=0}^{n-1} \mathrm{i}^k\frac{(-1)^{k+1}}{(k+1)!}\Big(\sum_{\alpha=1}^n\epsilon^\alpha \ad{A_\alpha} \Big)^k\left(\dot A_j\right) + \caO(\epsilon^{n+1}) = : \sum_{\alpha=1}^{n-1} \epsilon^\alpha Q_\alpha + \epsilon^{n}q_{n-1}(\epsilon).
\end{equation*}
Here,
\begin{equation}\label{Qa}
Q_\alpha = -\mathrm{i} \sum_{\mathbf j: s(\mathbf j) = \alpha}\frac{(-\mathrm{i})^{k}}{k!}\ad{A_{j_k}}\cdots\ad{A_{j_2}}(\dot A_{j_1}),
\end{equation}
namely
\begin{equation*}
Q_1 = -\dot A_1,\qquad Q_2 = -\dot A_2 + \frac{\iu}{2}[A_1,\dot A_1].
\end{equation*}
To proceed it is natural to define also $Q_0 = -K$. Altogether,
\begin{equation}\label{RExpansion}
\iu\epsilon  U_n\str \dot U_n - \epsilon K + H - U_n\str H U_n = \sum_{\alpha=1}^n \epsilon^\alpha (Q_{\alpha-1} - H_\alpha) + \epsilon^{n+1}(q_{n-1}(\epsilon) - h_{n}(\epsilon)).
\end{equation}
For $1\leq \alpha\leq n$, we observe immediately that $H_\alpha$ depends only on $\{A_j: 1\leq j\leq \alpha\}$, while $Q_\alpha$ depends on $\{A_j: 1\leq j\leq \alpha-1\}$ and $\{\dot A_j: 1\leq j\leq \alpha \}$.

We now inductively construct smooth local Hamiltonians $\{A_1,\ldots,A_n\}$, such that $A_\alpha\in\caL_{\caS,\infty}$ and $A^{(\tilde m-\alpha)}_\alpha\in\caL_{\caS,\infty}$, such that
\begin{equation}\label{orders}
[Q_{\alpha-1} - H_\alpha, P] = 0
\end{equation}
for $1\leq \alpha\leq n$. Defining then
\begin{equation*}
R_n := \epsilon^{n+1}U_n (q_{n-1}(\epsilon) - h_{n}(\epsilon))U_n\str,
\end{equation*}
equation~(\ref{Pi equation}) reduces to
\begin{equation*}
\iu\epsilon \dot\Pi_n = [H + R_n,\Pi_n]
\end{equation*}
as claimed. 

The case $\alpha = 1$. Here, (\ref{orders}) reads
\begin{equation}\label{eq:A1}
[K - \iu [A_1,H] , P] = 0
\end{equation}
which has solution $A_1\in\caL_{\caS,\infty}$ given by
\begin{equation*}
A_1 = \caI(K)
\end{equation*}
by Proposition~\ref{prop:HInverse}(ii). Moreover, by Assumption~\ref{A:Hamiltonians} and Proposition~\ref{prop:HInverse}(iii,iv), we have first that $K = \caI(\dot H) \in\caC^{(\tilde m-1)}$ implies that $A_1\in\caC^{(\tilde m-1)}$, and second that all derivatives of $K$ define local Hamiltonians, so that all derivatives of $A_1$ do so as well.

Induction. Let $\alpha>1$ and assume that the claim holds for all $1\leq \beta<\alpha$. Isolating the dependence on $A_\alpha$ in~(\ref{orders}) by writing
\begin{equation*}
H_\alpha= -\iu [A_\alpha,H] +  L_\alpha,
\end{equation*}
the equation becomes
\begin{equation}\label{AalphaEq}
[ Q_{\alpha-1} - L_\alpha, P] + \iu [ [A_\alpha,H], P] = 0.
\end{equation}
Note that $L_\alpha$ is a linear combination of multicommutators involving only $H$ and $\{A_\beta:1\leq \beta\leq \alpha-1\}$ by~(\ref{Ha}). Hence, by the induction hypothesis, $L_\alpha \in \caL_{\caS,\infty}$ by Lemma~\ref{lem:IntComm} and $L_\alpha\in\caC^{(\tilde m-\alpha+1)}$ with all derivatives defining local Hamiltonians. We now choose
\begin{equation}\label{Choice of Phi}
A_\alpha= \caI(L_\alpha - Q_{\alpha-1}),
\end{equation}
which is indeed a solution of~(\ref{AalphaEq}) by Proposition~\ref{prop:HInverse}(ii). By the remarks above and Proposition~\ref{prop:HInverse}(iii), $A_\alpha\in\caL_{\caS,\infty}$. Finally, $Q_\alpha\in\caC^{(\tilde m-\alpha-1)}$ by~(\ref{Qa}) and $L_\alpha\in\caC^{(\tilde m-\alpha+1)}$ imply that $A_\alpha\in\caC^{(\tilde m-\alpha)}$, see Proposition~\ref{prop:HInverse}(iv).

It remains to prove the local estimates on $R_n$. First of all,
\begin{equation*}
\epsilon^{n+1}h_n = \sum_{\substack{{\bf j}: s({\bf j})>n \\ j_i\leq n,k\leq n}}\frac{(-\iu)^k}{k!}\epsilon^{s({\bf j})}\ad{A{j_k}}\cdots \ad{A_{j_1}}(H) 
+ (-\iu)^{n+1}\int_{\Sigma_{n+1}}\ep{\iu u_{n+1}S_n}\ad{S_{n}}^{n+1}(H)\ep{-\iu u_{n+1}S_n} d{\bf u},
\end{equation*}
where $\Sigma_{n+1} =\{{\bf u}\in\bbR^{n+1}:0\leq u_1\leq \ldots \leq u_{n+1}\leq 1\}$. Note that the sum has a finite number of terms, while both terms are of order $\epsilon^{n+1}$. The multicommutators can be estimated by Lemma~\ref{lem:IntComm}(ii), while the action of $\exp(-\iu \ad{S_n})$ is controlled by Lemma~\ref{lem:IofPhi}(i) -- in fact a slightly simpler version thereof where $W_\gamma$ is replaced by the indicator function of $[0,1]$. Hence,
\begin{multline}\label{hn}
\epsilon^{n+1} \Vert h_n \Vert_{\caS,l} \leq \sum_{\substack{{\bf j}: s({\bf j})>n \\ j_i\leq n,k\leq n}}\frac{\epsilon^{s({\bf j})}}{k!} C^k 2^{k(l+k)}\Vert \Phi_{A_{j_k}} \Vert_{\caS,l+k} \cdots \Vert \Phi_{A_{j_1}} \Vert_{\caS,l+k}\Vert \Phi_{H} \Vert_{\caS,l+k} \\
+\frac{1}{(n+1)!} C^{n+1} 2^{(n+1)(l+n+2)} \Vert \Phi_{S_n}\Vert_{\caS,l+n+2}^{n+1}\Vert \Phi_{H} \Vert_{\caS,l+n+2},
\end{multline}
for any $l\in\bbN\cup\{0\}$. This yields an upper bound on $\Vert h_n \Vert_{\caS,l}$ that is uniform in $\epsilon$ for $0<\epsilon\leq 1$, since $s({\bf j})\geq n+1$ and  $\Vert \Phi_{S_n}\Vert_{\caS,l+n+2}$ is of order $\epsilon$. A similar bound holds for $q_{n-1}$ with $H$ being replaced by $\dot A_\alpha$, concluding the proof.
\end{proof}

As it is to be expected in adiabatic theory, the $A_\alpha$ depend locally in time on the Hamiltonian and its derivatives.
\begin{lemma}\label{lma:local}
With the assumptions of Lemma~\ref{ConstructPhi}, if there is $s_0\in[0,1]$ such that $H_{s_0}^{(\alpha)} = 0$ for all $1\leq \alpha \leq k \leq d+m$, then $A_\alpha(s_0) = 0$ for all $1\leq \alpha\leq k$. 
\end{lemma}
\begin{proof}
Since $\dot{H}_{s_0} = 0$, it follows from~(\ref{diff of I}) that $\frac{d}{ds}  \caI_s(G_s)\vert_{s=s_0} = \caI_{s_0}(\dot{G}_{s_0})$ for all differentiable $G_s$.  Using this $j-1$ times to differentiate Eq.~(\ref{Choice of Phi}), we see that $A_{\alpha}(s_0) = \dots = A_{\alpha}^{(j-1)}(s_0) = 0$ provided $A_1^{(j')}(s_0) = \dots = A_{\alpha-1}^{(j')}(s_0) = 0$ for all $0 \leq j' \leq j$. The result is then proved by recursion by realizing that the first $k-1$-derivatives of  $A_1(s) = \caI_{s}(K_{s})$ vanish at $s_0$. This holds  because 
$K_{s_0}^{(\alpha-1)} = \caI_{s_0}(H_{s_0}^{(\alpha)})$, and the corresponding derivatives of the Hamiltonian vanish by assumptions of the lemma.
\end{proof}

\begin{proof}[Proof of Theorem~\ref{thm:mainVector}] In this proof, we reinstate temporarily the $s$-dependence everywhere, with the conventions used in Section~\ref{sub:Setup}. We first note that by the assumptions and Lemma~\ref{lma:local}, $U_{n,\epsilon}(0) = \id$. Let $V_{n,\epsilon}(s,s')$ be the solution of
\begin{equation*}
\mathrm{i}\epsilon\frac{d}{ds} V_{n,\epsilon}(s,s') = (H_s + R_{n,\epsilon}(s) ) V_{n,\epsilon}(s,s'),\qquad V_{n,\epsilon}(s,s') = \id.
\end{equation*}
Since $V_{n,\epsilon}(s,0)P_0V_{n,\epsilon}(s,0)\str$ solves~(\ref{PiDynamics}) for the same initial condition, we have that 
\begin{equation}\label{UnVn}
V_{n,\epsilon}(s,0)P_0V_{n,\epsilon}(s,0)\str  =  \Pi_{n,\epsilon}(s) = U_{n,\epsilon}(s)P_s U_{n,\epsilon}(s)\str,
\end{equation}
for all $s\in[0,1]$. In order to compare it to the Schr\"odinger evolution
\begin{equation*}
\iu \epsilon \frac{d}{ds} U_\epsilon(s,s') = H_s U_\epsilon(s,s'),
\end{equation*}
 we note that for any operator $O$,
\begin{multline}\label{UV}
V_{n,\epsilon}(s,s')^* O V_{n,\epsilon}(s,s') - U_\epsilon(s,s')^* OU_\epsilon(s,s') = V_{n,\epsilon}(r,s')^* U_\epsilon(s,r)^* O U_\epsilon(s,r)V_{n,\epsilon}(r,s')\vert_{r=s'}^{r=s} \\
= \frac{-\mathrm{i}}{\epsilon}\int_{s'}^s V_{n,\epsilon}(r,s')^* \left[R_{n,\epsilon}(r), U_\epsilon(s,r)^* O U_\epsilon(s,r) \right] V_{n,\epsilon}(r,s') dr.
\end{multline}

Let now $\psi_\epsilon(s)$ be the solution of Schr\"odinger's equation, namely
\begin{equation*}
\psi_\epsilon(s) = U_\epsilon(s,0)\psi_0,
\end{equation*}
with the initial condition $\psi_0$ in the range of $P_0$, see~(\ref{Initial condition}). Let also $\phi_{n,\epsilon}(s)$ be the solution of Schr\"odinger's equation with counterdiabatic driving for the same initial condition,
\begin{equation*}
\phi_{n,\epsilon}(s) = V_{n,\epsilon}(s,0)\psi_0.
\end{equation*}
Then by (\ref{UV}),
\begin{equation}\label{Duhamel 1}
\langle \psi_\epsilon(s), O\psi_\epsilon(s)\rangle - \langle \phi_{n,\epsilon}(s), O \phi_{n,\epsilon}(s)\rangle = \frac{\iu}{\epsilon}\int_0^s \langle \phi_{n,\epsilon}(r), \left[R_{n,\epsilon}(r), U_\epsilon(s,r)^* O U_\epsilon(s,r) \right]\phi_{n,\epsilon}(r)\rangle dr.
\end{equation}
Now, by Lemma~\ref{ConstructPhi}, we have that $R_{n,\epsilon}(r)$ is a local Hamiltonian of order $\epsilon^{n+1}$. Furthermore, its commutator with a local observable $O$ evolved for a time $\vert \epsilon^{-1} s-\epsilon^{-1} r \vert$ can be controlled using Lemma~\ref{lma:rest}, yielding
\begin{equation*}
\left\vert \langle \phi_{n,\epsilon}(r), \left[R_{n,\epsilon}(r), U_\epsilon(s,r)^* O U_\epsilon(s,r) \right]\phi_{n,\epsilon}(r)\rangle \right\vert
\leq C r_{n,0}(r) \vert \supp(O)\vert^2 \Vert  O \Vert \epsilon^{n +1 - d}.
\end{equation*}
Finally, we note that by~(\ref{UnVn}), 
\begin{equation*}
\phi_{n,\epsilon}(s) = V_{n,\epsilon}(s,0)\psi_0 \in \mathrm{Ran} (\Pi_{n,\epsilon}(s)) = \mathrm{Ran} \left(U_{n,\epsilon}(s)P_sU_{n,\epsilon}(s)\str\right),
\end{equation*}
namely, there is a vector $\widetilde\psi_{n,\epsilon}(s) \in \mathrm{Ran}P_s$ such that $\phi_{n,\epsilon}(s) = U_{n,\epsilon}(s)\widetilde\psi_{n,\epsilon}(s)$, whence
\begin{equation*}
\langle \phi_{n,\epsilon}(s), O \phi_{n,\epsilon}(s)\rangle - \langle \widetilde\psi_{n,\epsilon}(s), O  \widetilde\psi_{n,\epsilon}(s)\rangle = \langle \widetilde\psi_{n,\epsilon}(s),  (U_{n,\epsilon}(s)\str O U_{n,\epsilon}(s) -O) \widetilde\psi_{n,\epsilon}(s) \rangle.
\end{equation*}
By Duhamel's formula
\begin{equation}\label{Duhamel 2}
 U_{n,\epsilon}(s)\str O U_{n,\epsilon}(s) -O = \iu\int_0^1\int_0^1\left[\ep{-\iu S_n(s')}O\ep{\iu S_n(s')}, \ep{-\iu uS_n(s')}\dot S_n(s') \ep{\iu uS_n(s')}\right]d u ds',
\end{equation}
which is bounded again by Lemma~\ref{lma:rest} and the fact that $\dot S_n$ is of order $\epsilon$. Altogether, this yields the bound
\begin{equation*}
\left\vert \langle \psi_\epsilon(s), O\psi_\epsilon(s)\rangle - \langle  \widetilde\psi_{n,\epsilon}(s), O  \widetilde\psi_{n,\epsilon}(s)\rangle \right\vert \leq \epsilon \Vert  O \Vert \vert \supp(O)\vert^2 \left(C_1\epsilon^{n-(d+1)} + C_2\right) 
\end{equation*}
which is the claim (i) of the theorem we had set to prove with the choice $n \geq d+1$.

Under the stronger smoothness assumptions, we can choose $n = d+m$, yielding an estimate of order $\epsilon^{m}$ for~(\ref{Duhamel 1}). Since the driving has stopped as $s=1$, Lemma~\ref{lma:local} implies that $A_\alpha(1) =0$ for $1\leq \alpha\leq m+d$. Hence $S_{m+d,{\epsilon}}(1) = 0$ and the claim (ii) of the theorem holds with $\widetilde\psi_{n,\epsilon}(s)= \phi_{n,\varepsilon}$.
\end{proof}
We are now in a position to discuss more explicitly the possible choices mentioned on page~\pageref{eq: parallel transport} for the vector $\widetilde\psi_{\epsilon}(s)$ of the theorem. The proof above yields a concrete expression for it, namely
\begin{equation*}
\widetilde\psi_{n,\epsilon}(s) = U_{n,\epsilon}(s)\str\phi_{n,\epsilon}(s)  = U_{n,\epsilon}(s)\str V_{n,\epsilon}(s,0)\psi_0.
\end{equation*}
Hence, $\widetilde\psi_{n,\epsilon}(s)$ is a solution of the differential equation
\begin{equation} \label{eq: equation tilde psi}
\iu \frac{d}{ds}\widetilde{\psi}_{n,\epsilon}(s) = \widetilde K_{n,\epsilon}(s) \widetilde{\psi}_{n,\epsilon}(s),\qquad \widetilde{\psi}_{n,\epsilon}(0) = \psi_0,
\end{equation}
where
\begin{equation}\label{eq:Ktilde}
\widetilde K_{n,\epsilon}(s) = \iu \dot U_{n,\epsilon}(s)\str U_{n,\epsilon}(s) + U_{n,\epsilon}(s)\str \epsilon^{-1}(H_s + R_{n,\epsilon}(s) U_{n,\epsilon}(s).
\end{equation}
Expanding in powers of $\epsilon$, we get
\begin{equation*}
\widetilde K_{n,\epsilon}(s) = \epsilon^{-1}H_s - \iu [A_1(s),H_s]  + B(s),
\end{equation*}
where again $B\in\caL_{\caS,\infty}$ is of order $\epsilon$. This expansion may seem useless since the leading term $ \epsilon^{-1}H_s$ is simply the generator of the original Schr\"odinger equation. However, we can capitalize on the knowledge that $\widetilde{\psi}_{n,\epsilon}(s)=P_s\widetilde{\psi}_{n,\epsilon}(s)$ to multiply any term in $\widetilde K_{n,\epsilon}(s)$ from the right with $P_s$. Let us for simplicity assume first exact degeneracy in the patch $P_s$, i.e.\ $\delta=0$  and set, without loss of generality, $H_sP_s=0$. Then we conclude that the highest order term vanishes, while the $\epsilon$-independent term reduces to
\begin{equation*}
\iu [A_1(s),H_s]P_s=  [\iu[A_1(s),H_s],P_s]P_s,
\end{equation*}
which is also equal to $[K_s,P_s]P_s$ by~(\ref{eq:A1}), and further to $-\iu \dot P_sP_s = \iu[P_s,\dot P_s]P_s$. 

Since $\widetilde K_{n,\epsilon}(s)$ and $\iu[P_s,\dot P_s]$ differ by a local Hamiltonian, $B(s)$, of order $\epsilon$ when acting on the range of $P_s$, and both generate a dynamics within $\mathrm{Ran} P_s$, the Duhamel formula yields
\begin{equation*}
\sup_{s\in[0,1]}\left\vert
\langle \widetilde \psi_{n,\epsilon}(s), O \widetilde \psi_{n,\epsilon}(s) \rangle - \langle \Omega(s), O \Omega(s)\rangle
\right\vert \leq C(O) \epsilon
\end{equation*}
where $\Omega(s)$ is the solution of the parallel transport equation \eqref{eq: parallel transport 2}.

It remains to explain that this conclusion remains valid when  the strict degeneracy $\delta=0$ is relaxed to the condition \eqref{eq: condition delta}. In that case, the expansion for  $\widetilde K_{n,\epsilon}(s)P_s$ has two additional terms, namely
$$
\epsilon^{-1} P_sH_sP_s,\qquad \text{and} \qquad \iu P_s[A_1(s),H_s] P_s.
$$
The first term has a norm bound $\epsilon^{-1} \delta$, and the second term $C\delta |\Lambda| $.  Therefore, if \eqref{eq: condition delta} holds, both terms have norm bounded by $C\epsilon$ and the argument goes through.

We now turn to linear response theory.
\begin{proof}[Proof of Theorem~\ref{thm:LR}]
Let $H_{s,\alpha} = H_{\mathrm{initial}} + \ep{s} \alpha V$, and $P_{\epsilon,\alpha}(s)$ be the solution of~(\ref{SchDensity}) with initial condition $\lim_{s\to-\infty}P_{\epsilon,\alpha}(s) = P_\alpha$. The proof of Lemma~\ref{ConstructPhi} can be reproduced with the estimates $r_{n,k}(s)$ on the counterdiabatic driving satisfying $r_{n,k}(s) = C_{n,k} \alpha \ep s$, since $ H_{s,\alpha}^{(k)} = \ep{s} \alpha V$ for all $k>1$. Carrying on with the proof of the theorem, we obtain instead of~(\ref{Duhamel 1})
\begin{equation*}
\left\vert \frac{\Tr( P_{\epsilon,\alpha}(s) O)}{\Tr( P_{\epsilon,\alpha}(s))} - \frac{\Tr( \Pi_{n,\epsilon}(s) O)}{\Tr( \Pi_{n,\epsilon}(s))}\right\vert
\leq  C \alpha \vert X\vert^2 \Vert  O \Vert \epsilon^{n-d} \lim_{s_0\to-\infty}\int_{s_0}^s \ep{s'} \vert s'-s_0\vert^d ds',
\end{equation*}
which is finite thanks to the exponential factor. A similar argument holds for~(\ref{Duhamel 2}) with a similar conclusion, where the norm of $ S_n$ provides the factor $\alpha$. Altogether,
\begin{equation*}
\left\vert  \omega_{\epsilon,\alpha;\sigma}(O) -  \omega_\alpha(O)\right\vert\leq  C \alpha \vert X\vert^2 \Vert  O \Vert \epsilon
\end{equation*}
uniformly in $\Lambda,\sigma$. Hence,
\begin{equation*}
\alpha^{-1}\vert \omega_{\epsilon,\alpha;0}(J) - \omega_0(J) + \iu\alpha\omega_0([K_0,J])\vert 
\leq  C \vert \mathrm{supp}(J)\vert^2 \Vert  J \Vert \epsilon
+ \vert \alpha^{-1}( \omega_\alpha(J) -\omega_0(J)) + \iu\omega_0([K_0,J]) \vert .
\end{equation*}
Now, by Corollary~\ref{cor:quasi-ad}, $\frac{d}{d\alpha}\omega_\alpha(O)= -\iu\omega_\alpha([K_\alpha,O])$ and by Proposition~\ref{prop:HInverse}(iv), this derivative is continuous in $\alpha$, uniformly in the volume.  It follows that 
\begin{equation*}
 \alpha^{-1}( \omega_\alpha(J) -\omega_0(J)) + \iu\omega_0([K_0,J]) \longrightarrow 0,
\end{equation*}
as $\alpha\to 0$, uniformly in the volume.
\end{proof}


\begin{proof}[Proof of the equality \eqref{eq: standard expression kubo}]
We abbreviate $H_{\mathrm{initial}}$ by $H$ and the corresponding spectral projection $P_0$ by $P$. We manipulate the right hand side of \eqref{eq: standard expression kubo}. First, $V$ can be replaced by $\widetilde V =PV(1-P)+(1-P)VP$ without changing the value of $f_{J,V}$. Then, since $\widetilde V$ is off-diagonal in the sense of Proposition \ref{prop:HInverse}, we have 
\begin{equation} \label{eq: repeat i}
\widetilde V=-\iu [H,\mathcal{I}(\widetilde V)].
\end{equation} Furthermore, for any off-diagonal $A$, using the spectral decomposition of $H$,
\begin{equation} \label{eq: convergence to one}
\iu\int_{0}^{\infty} dt e^{-\delta t}  \tau_{-t}([H,A])=  \int \frac{\iu(\mu-\lambda)}{\iu (\mu-\lambda) +\delta}  \dd E(\lambda) A  \dd E(\mu)  \longrightarrow  A
\end{equation}
as $\delta\downarrow 0$. Indeed, the spectral gap assumption combined with the fact that $A$ is off-diagonal ensures that $|\mu-\lambda| \geq \gamma >0$. 
Choosing $A=\caI(\widetilde V)$ and using the  equalities (\ref{eq: repeat i},\ref{eq: convergence to one}),  we find that the right hand side of \eqref{eq: standard expression kubo} equals $
-\iu \omega_{0}\left( [\caI(\widetilde V),J]\right)$. Since we can here again change $\widetilde V$ to $V$, we recover indeed the form for $f_{J,V}$ given in Theorem \ref{thm:LR}.
\end{proof}
%

\subsection{Technical lemmas}

It remains to prove the few technical results used above.

\begin{lemma}\label{lem:IntComm}
With the definition~(\ref{Interaction Commutator}) of the interaction associated with a commutator of local Hamiltonians, the following holds:
\begin{enumerate}
\item If $H\in\caL_{\zeta,0}$, then for {any $O$ with $\mathrm{supp} (O) \subset \Lambda$},
\begin{equation*}
\Vert[H,O] \Vert \leq  2\Vert O\Vert \vert \mathrm{supp}(O) \vert \Vert F_\zeta\Vert_1 \Vert \Phi_H\Vert_{\zeta,0}.
\end{equation*}
\item Let $n,k\in\bbN$ and $A_0,\ldots, A_k\in \caL_{\zeta,n+k}$. Then $\ad{A_k}\cdots\ad{A_1}(A_0)\in\caL_{\zeta,n}$ and there is a $C>0$ depending on $\zeta$ but not on $n,k$, such that 
\begin{equation*}
\Vert\Phi_{\ad{A_{j_k}}\cdots \ad{A_{j_1}}(A_0) }\Vert_{\zeta,n}\leq  C^k 2^{k(n+k)}\Vert \Phi_{A_{j_k}} \Vert_{\zeta,n+k} \cdots \Vert \Phi_{A_{j_1}} \Vert_{\zeta,n+k}\Vert \Phi_{A_0} \Vert_{\zeta,n+k}.
\end{equation*}
\item If $A_0,\ldots, A_k\in \caL_{\zeta,\infty}$, then $\ad{A_k}\cdots\ad{A_1}(A_0)\in\caL_{\zeta,\infty}$
\end{enumerate}
\end{lemma}
\begin{proof}
Claim~(i) follows from
\begin{equation*}
\Vert[H,O]\Vert\leq 2 \Vert O \Vert\sup_{\Lambda\in\caF(\Gamma) } \sum_{x\in\mathrm{supp}O}\sum_{y\in\Lambda}\sum_{\substack{X\subset\Lambda :\\ x,y\in X}}\frac{\Vert \Phi^\Lambda_H(X)\Vert}{F_\zeta(d(x,y))}F_\zeta(d(x,y))
 \leq 2 \Vert O \Vert \vert \mathrm{supp}O\vert \Vert F_\zeta\Vert_1 \Vert \Phi_H\Vert_{\zeta,0}.
\end{equation*}
For Claim~(ii), we first recall that the interaction terms of $G = [A_1,A_0]$ are given by $[\Phi^\Lambda_{A_1}(X_1),\Phi^\Lambda_{A_0}(X_0)]$ whenever $X_0\cap X_1\neq \emptyset$. With $Z = X_0\cup X_1$ and hence
\begin{equation*}
\vert Z\vert^n \leq\sum_{k=0}^n\binom{n}{k}\vert X_0\vert^k\vert X_1\vert^{n-k},
\end{equation*}
we distinguish two contributions to the sum $\sum_{Z\ni \{x,y\}}\vert Z\vert^n\Vert \Phi_G^\Lambda(Z)\Vert/F_\zeta(d(x,y))$: Either $x,y\in X_0$ or $x\in X_0,y\in X_1\setminus X_0$. The first contribution can be bounded by
\begin{multline*}
\sum_{k=0}^n\binom{n}{k}\sum_{X_0\ni\{x,y\}}\vert X_0\vert^k\frac{\Vert\Phi^\Lambda_{A_0}(X_0)\Vert}{F_\zeta(d(x,y))}\sum_{z_0\in X_0}\sum_{z_1\in\Lambda}\sum_{X_1\ni\{z_0,z_1\}}\vert X_1\vert^{n-k}\frac{\Vert\Phi_{A_1}^\Lambda(X_1)\Vert}{F_\zeta(d(z_0,z_1))}F_\zeta(d(z_0,z_1)) \\
\leq \Vert F_\zeta\Vert_1 \sum_{k=0}^n \binom{n}{k} \Vert \Phi_{A_0}\Vert_{\zeta,k+1}\Vert \Phi_{A_1}\Vert_{\zeta,n-k}.
\end{multline*}
The second contribution is estimated as
\begin{multline*}
\sum_{k=0}^n\binom{n}{k}\sum_{z\in \Lambda}\frac{F_\zeta(d(z,y))F_\zeta(d(x,z)) }{F_\zeta(d(x,y))}\sum_{X_0\ni\{x,z\}}\vert X_0\vert^k\frac{\Vert\Phi^\Lambda_{A_0}(X_0)\Vert}{F_\zeta(d(x,z))}\sum_{X_1\ni\{z,y\}}\vert X_1\vert^{n-k}\frac{\Vert\Phi^\Lambda_{A_1}(X_1)\Vert}{F_\zeta(d(z,y))}\\
\leq C_\zeta \sum_{k=0}^n\binom{n}{k}\Vert \Phi_{A_0}\Vert_{\zeta,k}\Vert \Phi_{A_1}\Vert_{\zeta,n-k}.
\end{multline*}
Both estimates imply that if $A_0,A_1\in\caL_{\zeta,\infty}$, then $[A_1,A_0]\in\caL_{\zeta,\infty}$. More precisely, if $A_0,A_1\in\caL_{\zeta,n+1}$, then $[A_1,A_0]\in\caL_{\zeta,n}$ with
\begin{equation*}
\Vert\Phi_{[A_1,A_2]}\Vert_{\zeta,n}\leq C 2^n\Vert \Phi_{A_1} \Vert_{\zeta,n+1}\Vert \Phi_{A_2} \Vert_{\zeta,n+1},
\end{equation*}
where $C$ depends on $\zeta$ but not on $n$, and we used that $\Vert \Phi\Vert_{\zeta,n}\leq \Vert \Phi\Vert_{\zeta,m}$ whenever $n\leq m$. The result follows by induction. Finally, Claim~(iii) is an immediate consequence of~(ii).
\end{proof}
\begin{lemma}\label{lma:rest}
Let $H\in\caL_{\caE,0}$ generate the dynamics $\tau_{t,t'}$, and let $v$ be the corresponding Lieb-Robinson velocity~(\ref{LRv}). Then for {any $O$ with $\mathrm{supp}(O) \subset \Lambda$} and $A\in\caL_{\zeta,n}$, then there is a $C>0$ such that  
\begin{equation*}
\left\Vert \left[A , \tau_{t,t'}(O) \right]\right\Vert\leq  C \Vert \Phi_A\Vert_{\zeta,0} \vert \mathrm{supp}(O)\vert^2 |t-t'|^{d},
\end{equation*}
whenever $ |t-t'|\geq v^{-1}$. 
\end{lemma}
\begin{proof}
For any $Y \subset \Gamma$ and $n \geq 0$, denote by
\begin{equation*}
Y_n := \left\{ z \in \Gamma : d(z,Y) \leq n \right\}
\end{equation*}
the fattening of the set $Y$ by $n$. If $O\in\caA^X$, let
\begin{equation*}
O^0:= \Pi^{X_{v\delta t}}(\tau_{t,t'}(O)),
\end{equation*}
where $\delta t = |t-t'|$ and $\Pi^{Y}:\caA\to\caA^Y$ is the partial trace, and for any $k\geq 1$,
\begin{equation*}
O^k:= \Pi^{X_{v\delta t+k}}(\tau_{t,t'}(O)) - \Pi^{X_{v\delta t+(k-1)}}(\tau_{t,t'}(O)).
\end{equation*}
By construction, $\tau_{t,t'}(O) = \sum_{k=0}^\infty O^k$, where the sum is actually finite since the underlying $\Lambda$ is finite. Since $O^k$ is strictly local, Lemma~\ref{lem:IntComm}(i) implies that, {with $X=\supp(O)$},
\begin{equation*}
\Vert [A, O^k] \Vert \leq C \Vert \Phi_A\Vert_{\zeta,0} \Vert O^k \Vert \vert X \vert \delta t^d k^d
\end{equation*}
where we used that
\begin{equation*}
\left\vert X_{v\delta t+k} \right\vert \leq \vert X\vert \kappa (v\delta t+k)^d \leq \vert X\vert \kappa (2k v\delta t)^d
\end{equation*}
since the lattice is $d$-dimensional, see~(\ref{ddim}). In order to bound the norm of $O^k$, we add and subtract an identity and recall that
\begin{equation*}
\left\Vert \left(\Pi^{X_{v\delta t+k}}-\id\right)(\tau_{t,t'}(O))\right\Vert \leq C \Vert O \Vert \vert X \vert \ep{-\mu k}
\end{equation*}
by the Lieb-Robinson bound and~\cite{Nachtergaele:2013ts}. Note that $C$ depends on the decay rate of $\Phi_H$. Gathering these estimates, we finally obtain
\begin{equation*}
\Vert [A, \tau_{t,t'}(O)]\Vert\leq C \Vert \Phi_A\Vert_{\zeta,0} \Vert O \Vert \vert X \vert^2 \vert t-t'\vert^d\sum_{k=0}^\infty   k^d \ep{-\mu k}
\end{equation*}
which is the claim since the series converges.
\end{proof}
The last result pertaining to Lieb-Robinson bounds expresses the locality of the generator of the spectral flow, and it is an adapted version of Theorem~4.8 in~\cite{automorphic}:
\begin{lemma}\label{lem:IofPhi}
If Assumption~\ref{A:Hamiltonians} holds, then
\begin{enumerate}
\item if $B\in\caL_{\caS,k+1}$, then $\caI(B)\in\caL_{\caS,k}$, and in particular $B\in\caL_{\caS,\infty}$ implies $\caI(B)\in\caL_{\caS,\infty}$
\item the generator $K$ of the spectral flow belongs to $\caL_{\caS,\infty}$.
\end{enumerate}
\end{lemma}
\begin{proof}
(i) First, we note that by a slightly extended argument to the one above, for any $O \in \mathcal{A}^X$, there are operators $\mbox{supp} \left( \Delta_{n}(O) \right) \subset X_{n} \cap \Lambda$, such that
\begin{equation}\label{Deltas}
 \int_{- \infty}^{\infty} \tau_t^\Lambda(O) W_{\gamma}(t) \, dt = \sum_{n=0}^{\infty} \Delta_{n}^\Lambda(O),
\end{equation}
and there is a function $\tilde \zeta\in\caS$, independent of $\Lambda$ and $s$, such that
\begin{equation} \label{eq:Delbd}
\left\| \Delta_{n}^\Lambda(O) \right\| \leq \Vert O \Vert \vert \mathrm{supp}(O)\vert \tilde \zeta(n),
\end{equation}
see also Lemma~C.3 of~\cite{bachmann2016lieb}.

Now, let $B\in\caL_{\caS,\infty}$ with associated $\Phi_B\in\caB_{\zeta,n}$ for a $\zeta\in\caS$, and recall that
\begin{equation*}
\caI^\Lambda(B^\Lambda) = \sum_{Z \in\caF(\Gamma)} \int_{-\infty}^{\infty} W_{\gamma}(t) \tau^\Lambda_t \left( \Phi_B^\Lambda(Z) \right)  dt.
\end{equation*}
In the notation above, we define an interaction for $\caI(B)$ by
\begin{equation}\label{eq:defPsiLambda}
\Phi_{\caI(B)}^\Lambda(Z) := \sum_{n \geq 0} \sum_{\substack{Y \subset \Lambda: \\ Y_n =Z}} \Delta_{n}^\Lambda \left ( \Phi_B^\Lambda(Y)\right),
\end{equation}
which satisfies $\caI^\Lambda(B^\Lambda) = \sum_{Z\subset\Lambda}\Phi_{\caI(B)}^\Lambda(Z)$ by~(\ref{Deltas}), and $\Phi_{\caI(B)}^\Lambda(Z)\in\caA^Z$ by construction. It remains to check the decay condition. We decompose
\begin{equation*}
\sum_{\stackrel{Z \subset \Lambda:}{x,y \in Z}}\vert Z\vert^k \Vert \Phi_{\caI(B)}^\Lambda(Z) \Vert 
 \leq \sum_{\stackrel{Z \subset \Lambda:}{x,y \in Z}} \sum_{\stackrel{Y,n \geq 0:}{Y_n =Z}} \vert Z\vert^k \Vert \Delta_{n}^\Lambda  ( \Phi_B^\Lambda(Y)) \Vert 
  =  \sum_{Y \subset \Lambda} \sum_{n \geq 0} \chi ( x,y \in Y_n ) \vert Y_n \vert^k \Vert \Delta_{n}^\Lambda  ( \Phi_B^\Lambda(Y)) \Vert
\end{equation*}
where $\chi$ is the indicator function and $\sup_{s\in[0,1]}$ is implicit everywhere. For a $d$-dimensional lattice,
\begin{equation*}
\vert Y_n\vert\leq   \vert Y \vert \kappa  n^d.
\end{equation*}
Rearranging the terms above,
\begin{multline} \label{PotentialEstimate}
 \sum_{\stackrel{Z \subset \Lambda:}{x,y \in Z}}\vert Z\vert^k \Vert \Phi_{\caI(B)}^\Lambda(Z) \Vert 
 \leq \kappa^k \sum_{\stackrel{Y \subset \Lambda:}{x,y \in Y}} \vert Y\vert^k \sum_{n \geq 0} n^{ kd} \Vert \Delta^\Lambda_{n}  ( \Phi_B^\Lambda(Y)) \Vert \\
 \quad + \kappa^k \sum_{m=1}^{\infty} \sum_{\stackrel{Y \subset \Lambda:}{x,y \in Y_m}} \chi( \{x,y\} \cap Y_{m-1}^c \neq \emptyset )  \vert Y\vert^k \sum_{n \geq m} n^{ kd} \Vert \Delta_{n}^\Lambda  ( \Phi_B^\Lambda(Y)) \Vert = S_1 + S_2.
\end{multline}
Now, $S_1$ can be estimated using~(\ref{eq:Delbd})
\begin{equation*}
S_1\leq \kappa^k \sum_{n \geq 0} n^{ kd} \tilde \zeta (n) \sum_{\stackrel{Y \subset \Lambda:}{x,y \in Y}} \vert Y \vert^{k+1} \Vert  \Phi_B^\Lambda(Y)\Vert 
\leq C \Vert  \Phi_B \Vert_{\zeta,k+1} F_{\zeta}(d(x,y)).
\end{equation*}
For $S_2$, we note that
\begin{equation} \label{eq:overcount}
\sum_{\stackrel{Y \subset \Lambda:}{x,y \in Y_m}} \chi( \{x,y\} \cap Y_{m-1}^c \neq \emptyset ) \leq  \sum_{z_1 \in B_m(x)} \sum_{z_2 \in B_m(y)} \sum_{\stackrel{Y \subset \Lambda:}{z_1,z_2 \in Y}} 1.
\end{equation}
This and again~(\ref{eq:Delbd}) yield
\begin{align*}
S_2 &\leq \kappa^k \sum_{m=1}^{\infty}  \sum_{z_1 \in B_m(x)} \sum_{z_2 \in B_m(y)} \sum_{\stackrel{Y \subset \Lambda:}{z_1,z_2 \in Y}} \vert Y \vert^{k+1} \Vert  \Phi_B^\Lambda(Y)\Vert 
\sum_{n \geq m} n^{ kd}  \tilde \zeta(n) \\
&\leq \kappa^k \Vert  \Phi_B \Vert_{\zeta,k+1}  \sum_{m=1}^{\infty} \bigg(\sum_{n \geq m} n^{ kd}  \tilde \zeta(n)\bigg)  \sum_{z_1 \in B_m(x)} \sum_{z_2 \in B_m(y)} F_{\zeta}(d(z_1,z_2)).
\end{align*}
Let $m_0 = d(x,y)/4$. Then for $m\leq m_0$,
\begin{equation*}
d(x,y)\leq d(x,z_1) + d(z_1,z_2) + d(z_2,y) \leq 2m + d(z_1,z_2)
\end{equation*}
so that $d(z_1,z_2)\geq d(x,y)/2$. We split the sum over $m_0$ in the bound for $S_2$, to obtain
\begin{equation*}
S_2   \leq C \Vert  \Phi_B \Vert_{\zeta,k+1}   F_{\zeta}(d(x,y)/2)  \sum_{m=1}^{m_0} m^{2d} +C \Vert  \Phi_B \Vert_{\zeta,k+1}  \sum_{m > m_0} m^d  \bigg(\sum_{n \geq m} n^{ kd}  \tilde\zeta(n)\bigg).
\end{equation*}
Here we wrote $C$ for coefficients that do not depend on $d(x,y)$ (or $\Lambda$), we used that $\tilde\zeta$ decays faster than any polynomial and that $F_{\zeta}$ is a decreasing function. 
Inspecting now the bounds for $S_1,S_2$ and using the fast decay of $\zeta$, we see that they can be cast in the form. 
\begin{equation*}
\sum_{\stackrel{Z \subset \Lambda:}{x,y \in Z}}\vert Z\vert^k \Vert \Phi_{\caI(B)}^\Lambda(Z) \Vert  \leq S_1+S_2 \leq  C \Vert  \Phi_B \Vert_{\zeta,k+1} F(d(x,y)) h(d(x,y)),
\end{equation*}
where $h$ is bounded, nonincreasing,  decays faster than any inverse power and it can be chosen such that  $h<1$ (by adjusting $C$). It remains to find a  $\xi'\in\caS$ such that $h\leq \xi'$ to conclude the proof. But the existence of such a function is guaranteed by Lemma~\ref{lma:sub} below, taking $f=-\log h$ and $\xi'=\exp(-\hat f)$.

(ii) The second statement follows immediately from~(i) since $K = \caI(\dot H)$ and $\dot H\in\caL_{\caE,\infty}$ by assumption.
\end{proof}

The above proof requires the following lemma.  A function $g$ on the positive reals $\mathbb{R}^+$ is called subadditive if 
$
g(x+y) \leq g(x) +g(y) 
$ for all $x,y$. 
Then 
\begin{lemma}\label{lma:sub}
Given a nondecreasing function $f$ on $\mathbb{R}^+$ for which $f(0)>0$, there is a subadditive, positive function $\hat f$ such that 
$$
\hat f \leq f.
$$ 
Moreover, if $f$ has the property that $f(x)-m\log x \to +\infty$ for any $m>0$, then $\hat{f}$ has this property as well. 
\end{lemma}
\begin{proof}
Set, following  \cite{bruckner1960minimal},
$$
\hat f(x) := \inf_{(x_i): \sum_i x_i=x}   \, \,  \sum_i  f(x_i)  
$$
Then by straightforward considerations, one verifies
\begin{enumerate}
\item[$a)$]  $0 < \hat f \leq f$,
\item[$b)$] $\hat f$ is subadditive, 
\item[$c)$] If $g$ is subadditive and $g \leq f$, then $g \leq \hat{f}$.
\end{enumerate}
Using properties of $f$, for any $m$, we can find $y_m\geq \mathrm{e}$ large enough so that 
\begin{enumerate}
\item $f(x) \geq m\log(x) $ for $x \geq y_m$,
\item $f(x)\geq mx/y_m$ for $x<y_m$. 
\end{enumerate}
Now consider the function
$$
g_m(x) :=\begin{cases}  (\tfrac{m}{y_m}) x  \qquad & \text{for}\,    x \leq y_m \\    
m\log x-m\log y_m +m & \text{for}\,  x \geq y_m \end{cases}
$$
This function has been constructed so as to be of class $C^1$, concave (hence subadditive), and to satisfy $ g_m \leq f$. 
From item $c)$ it then follows that $\hat{f} \geq g_m$.   Since this holds for any $m$, the second claim is proven.

\end{proof}

\section*{Acknowledgements}
The authors would like to thank Gian Michele Graf and Vojkan Jak\v{s}i\'c for insightful discussions. We also thank S.~Teufel for his comments on the first version of the article that helped us improved subsequent versions.
WDR is supported by the Flemish Research Fund (FWO).


\begin{thebibliography}{10}

\bibitem{BornFock}
M.~Born and V.~Fock.
\newblock {Beweis des Adiabatensatzes}.
\newblock {\em Zeitschrift f{\"u}r Physik}, 51(3-4):165--180, 1928.

\bibitem{Kato50}
T.~Kato.
\newblock On the adiabatic theorem of quantum mechanics.
\newblock {\em J. Phys. Soc. Japan}, 5:435--439, 1950.

\bibitem{kasuga1961adiabatic}
T.~Kasuga.
\newblock {On the adiabatic theorem for the Hamiltonian system of differential
  equations in the classical mechanics. I}.
\newblock {\em Proc. Japan Academy}, 37(7):366--371, 1961.

\bibitem{Teufel}
S.~Teufel.
\newblock {\em Adiabatic Perturbation Theory in Quantum Dynamics}.
\newblock Lecture Notes in Matematics. Springer, 2003.

\bibitem{Panati2003}
G.~Panati, H.~Spohn, and S.~Teufel.
\newblock {Effective dynamics for Bloch electrons: Peierls substitution and
  beyond}.
\newblock {\em Commun. Math. Phys.}, 242(3):547--578, 2003.

\bibitem{AE99}
J.E. Avron and A.~Elgart.
\newblock Adiabatic theorem without a gap condition.
\newblock {\em Commun. Math. Phys.}, 203:445--463, 1999.

\bibitem{Jansen}
S.~Jansen, M.-B. Ruskai, and R.~Seiler.
\newblock Bounds for the adiabatic approximation with applications to quantum
  computation.
\newblock {\em arXiv quant-ph/0603175}, 2006.

\bibitem{lorenz2005adiabatic}
K.~Lorenz, T.~Jahnke, and C.~Lubich.
\newblock Adiabatic integrators for highly oscillatory second-order linear
  differential equations with time-varying eigendecomposition.
\newblock {\em BIT Numerical Mathematics}, 45(1):91--115, 2005.

\bibitem{bradford2011adiabatic}
K.~Bradford and Y.~Kovchegov.
\newblock Adiabatic times for {Markov} chains and applications.
\newblock {\em J. Stat. Phys.}, 143(5):955--969, 2011.

\bibitem{gang2015adiabatic}
G.~Zhou and P.~Grech.
\newblock An adiabatic theorem for the {Gross-Pitaevskii} equation.
\newblock {\em Commun. Part. Diff. Eq.}, 2017.

\bibitem{AFGG}
J.E. Avron, M.~Fraas, G.M. Graf, and P.~Grech.
\newblock Adiabatic theorems for generators of contracting evolutions.
\newblock {\em Commun. Math. Phys.}, 314(1):163--191, 2012.

\bibitem{Broer:2004vs}
H.~Broer.
\newblock {KAM theory: the legacy of Kolmogorov's 1954 paper}.
\newblock {\em Bull. Amer. Math. Soc.}, 2004.

\bibitem{Kubo:1957cl}
R.~Kubo.
\newblock Statistical-mechanical theory of irreversible processes. {I. General}
  theory and simple applications to magnetic and conduction problems.
\newblock {\em J. Phys. Soc. Japan}, 12(6):570--586, 1957.

\bibitem{Simon:1984aa}
B.~Simon.
\newblock {\em Fifteen problems in mathematical physics}.
\newblock Perspectives in Mathematics. Birkh\"auser Verlag, Basel, 1984.

\bibitem{Bratteli:1997aa}
O.~Bratteli and D.W. Robinson.
\newblock {\em Operator Algebras and Quantum Statistical Mechanics 2:
  Equilibrium States. Models in Quantum Statistical Mechanics}.
\newblock Springer, 2nd edition, 1997.

\bibitem{Berry90}
M.V. Berry.
\newblock Histories of adiabatic quantum transitions.
\newblock {\em Proc. Royal Soc. London A}, 429(1876):61--72, 1990.

\bibitem{Nenciu}
G.~Nenciu.
\newblock {Linear adiabatic theory. Exponential estimates}.
\newblock {\em Commun. Math. Phys.}, 152(3):479--496, 1993.

\bibitem{Garrido}
L.M. Garrido.
\newblock Generalized adiabatic invariance.
\newblock {\em J. Math. Phys.}, 5(3):355--362, 1964.

\bibitem{Hagedorn}
G.A. Hagedorn and A.~Joye.
\newblock Elementary exponential error estimates for the adiabatic
  approximation.
\newblock {\em J. Math. Anal. and Applications}, 267(1):235--246, 2002.

\bibitem{TeufelAd}
D.~Monaco and S.~Teufel.
\newblock {Adiabatic currents for interacting electrons on a lattice}.
\newblock {\em arXiv math-ph/1707.01852v1}, 2017.

\bibitem{Cubitt:2015ch}
T.S. Cubitt, D.~P{\'e}rez-Garc{\'\i}a, and M.M Wolf.
\newblock {Undecidability of the spectral gap}.
\newblock {\em Nature}, 528(7581):207--211, 2015.

\bibitem{Nachtergaele:1996vc}
B.~Nachtergaele.
\newblock {The spectral gap for some spin chains with discrete symmetry
  breaking}.
\newblock {\em Commun. Math. Phys.}, 175(3):565--606, 1996.

\bibitem{yarotsky2004perturbations}
D.A. Yarotsky.
\newblock Perturbations of ground states in weakly interacting quantum spin
  systems.
\newblock {\em J. Math. Phys.}, 45(6):2134--2152, 2004.

\bibitem{AKLT}
I.~Affleck, T.~Kennedy, E.H. Lieb, and H.~Tasaki.
\newblock Valence bond ground states in isotropic quantum antiferromagnets.
\newblock {\em Commun. Math. Phys.}, 115:477--528, 1988.

\bibitem{Bravyi:2011ea}
S.~Bravyi and M.B. Hastings.
\newblock A short proof of stability of topological order under local
  perturbations.
\newblock {\em Commun. Math. Phys.}, 307(3):609--627, 2011.

\bibitem{Michalakis:2013gh}
S.~Michalakis and J.P. Zwolak.
\newblock Stability of frustration-free {Hamiltonians}.
\newblock {\em Commun. Math. Phys.}, 322(2):277--302, 2013.

\bibitem{Szehr:2015fn}
O.~Szehr and M.M. Wolf.
\newblock Perturbation theory for parent {Hamiltonians of Matrix Product
  States}.
\newblock {\em Journal of statistical physics}, 159(4):752--771, 2015.

\bibitem{Kitaev:2003ul}
A.Y. Kitaev.
\newblock {Fault-tolerant quantum computation by anyons}.
\newblock {\em Ann. Phys.}, 303(1):2--30, 2003.

\bibitem{yarotsky2004quasi}
D.A. Yarotsky.
\newblock Quasi-particles in weak perturbations of non-interacting quantum
  lattice systems.
\newblock {\em arXiv preprint math-ph/0411042}, 2004.

\bibitem{bachmann2016lieb}
S.~Bachmann, W.~Dybalski, and P.~Naaijkens.
\newblock {Lieb--Robinson bounds, Arveson spectrum and Haag--Ruelle scattering
  theory for gapped quantum spin systems}.
\newblock {\em Annales Henri Poincar{\'e}}, 17(7):1737--1791, 2016.

\bibitem{Sachdev:2000aa}
S.~Sachdev.
\newblock {\em Quantum Phase Transitions}.
\newblock Cambridge University Press, 2000.

\bibitem{Wen}
X.~Chen, Z.-C. Gu, and X.-G. Wen.
\newblock Local unitary transformation, long-range quantum entanglement, wave
  function renormalization, and topological order.
\newblock {\em Phys. Rev. B}, 82(15):155138, 2010.

\bibitem{automorphic}
S.~Bachmann, S.~Michalakis, S.~Nachtergaele, and R.~Sims.
\newblock Automorphic equivalence within gapped phases of quantum lattice
  systems.
\newblock {\em Commun. Math. Phys.}, 309(3):835--871, 2012.

\bibitem{BachmannOgata}
S.~Bachmann and Y.~Ogata.
\newblock {$C^1$-classification} of gapped parent {Hamiltonians} of quantum
  spin chains.
\newblock {\em Commun. Math. Phys.}, 338(3):1011--1042, 2015.

\bibitem{PVBS}
S.~Bachmann and B.~Nachtergaele.
\newblock Product vacua with boundary states and the classification of gapped
  phases.
\newblock {\em Commun. Math. Phys.}, 329(2):509--544, 2014.

\bibitem{HastingsWen}
M.B. Hastings and X.-G. Wen.
\newblock Quasiadiabatic continuation of quantum states: The stability of
  topological ground-state degeneracy and emergent gauge invariance.
\newblock {\em Phys. Rev. B}, 72(4):045141, 2005.

\bibitem{Giuliani:2016gn}
A.~Giuliani, V.~Mastropietro, and M.~Porta.
\newblock Universality of the {Hall} conductivity in interacting electron
  systems.
\newblock {\em Commun. Math. Phys.}, 2016.

\bibitem{HastingsMichalakis}
M.B. Hastings and S.~Michalakis.
\newblock Quantization of {Hall} conductance for interacting electrons on a
  torus.
\newblock {\em Commun. Math. Phys.}, 334:433--471, 2015.

\bibitem{Elgart:2004wd}
A.~Elgart and B.~Schlein.
\newblock {Adiabatic charge transport and the Kubo formula for Landau-type
  Hamiltonians}.
\newblock {\em Commun. Pure Appl. Math.}, 57(5):590--615, 2004.

\bibitem{LRLocal}
J.-M. Bouclet, F.~Germinet, A.~Klein, and J.H. Schenker.
\newblock Linear response theory for magnetic {Schr\"odinger} operators in
  disordered media.
\newblock {\em J. Funct. Anal.}, 226:301--372, 2005.

\bibitem{klein2007mott}
A.~Klein, O.~Lenoble, and P.~M{\"u}ller.
\newblock {On Mott's formula for the ac-conductivity in the Anderson model}.
\newblock {\em Annals of Math.}, pages 549--577, 2007.

\bibitem{Bru:2017aa}
J.-B. Bru and W.~de~Siqueira~Pedra.
\newblock {\em Lieb-Robinson Bounds for Multi-Commutators and Applications to
  Response Theory}.
\newblock SpringerBriefs in Mathematical Physics. Springer, 2017.

\bibitem{AbouSalem:2005kr}
W.K. Abou-Salem and J.~Fr{\"o}hlich.
\newblock Adiabatic theorems and reversible isothermal processes.
\newblock {\em Lett. Math. Phys.}, 72(2):153--163, 2005.

\bibitem{Jaksic:2006fv}
V.~Jak\v{s}ic, Y.~Ogata, and C.-A. Pillet.
\newblock The {Green-Kubo} formula and the {Onsager} reciprocity relations in
  quantum statistical mechanics.
\newblock {\em Commun. Math. Phys.}, 265(3):721--738, 2006.

\bibitem{Kampen:1971aa}
N.G. van Kampen.
\newblock The case against linear response theory.
\newblock {\em Phys. Norv.}, 5, 1971.

\bibitem{BratteliRobinsonBook}
O.~Bratteli and D.W. Robinson.
\newblock {\em Operator Algebras and Quantum Statistical Mechanics 1: {$C^*-$
  and $W^*-$Algebras, Symmetry Groups, Decomposition of States}}.
\newblock Springer, 2nd edition, 1987.

\bibitem{Nachtergaele:2006bh}
B.~Nachtergaele, Y.~Ogata, and R.~Sims.
\newblock Propagation of correlations in quantum lattice systems.
\newblock {\em J. Stat. Phys.}, 124(1):1--13, 2006.

\bibitem{Lieb:1972ts}
E.H. Lieb and D.W. Robinson.
\newblock {The finite group velocity of quantum spin systems}.
\newblock {\em Commun. Math. Phys.}, 28(3):251--257, 1972.

\bibitem{Hastings:2004go}
M.B. Hastings.
\newblock {Locality in Quantum and Markov Dynamics on Lattices and Networks}.
\newblock {\em Phys. Rev. Lett.}, 93(14):140402, 2004.

\bibitem{Osborne}
T.J. Osborne.
\newblock Simulating adiabatic evolution of gapped spin systems.
\newblock {\em Phys. Rev. A}, 75(3):032321, 2007.

\bibitem{Nachtergaele:2013ts}
B.~Nachtergaele, V.B. Scholz, and R.F. Werner.
\newblock {Local approximation of observables and commutator bounds}.
\newblock In J.~Janas, P.~Kurasov, A.~Laptev, and S.~Naboko, editors, {\em
  Operator Methods in Mathematical Physics: Conference on Operator Theory,
  Analysis and Mathematical Physics (OTAMP) 2010, Bedlewo, Poland}, 2013.

\bibitem{bruckner1960minimal}
Andrew Bruckner.
\newblock Minimal superadditive extensions of superadditive functions.
\newblock {\em Pacific Journal of Mathematics}, 10(4):1155--1162, 1960.

\end{thebibliography}
\end{document}